\documentclass[11pt]{article}

\usepackage[T1]{fontenc}
\usepackage{libertinus}
\usepackage{libertinust1math}
\usepackage[narrow,varqu,varl,scaled=0.95]{zi4}

\usepackage{microtype}
\usepackage[a4paper,margin=1.08in]{geometry}
\usepackage{amssymb,amsthm,mathtools}
\DeclareSymbolFont{wideaccents}{OMX}{lmex}{m}{n}
\DeclareMathAccent{\widehat}{\mathord}{wideaccents}{"62}
\DeclareMathAccent{\widetilde}{\mathord}{wideaccents}{"65}
\usepackage{enumitem}
\usepackage{xcolor}
\usepackage{authblk}
\usepackage[giveninits=true,maxbibnames=99,style=alphabetic,maxalphanames=4,minalphanames=3,isbn=false,maxcitenames=99]{biblatex}
\usepackage[colorlinks=true,linkcolor=blue!55!black,citecolor=blue!55!black,urlcolor=blue!55!black,hypertexnames=false]{hyperref}
\usepackage[nameinlink,capitalise]{cleveref}

\AtEveryBibitem{%
  \clearlist{language}%
  \clearlist{location}
}
\AtBeginBibliography{%
  \setlength{\emergencystretch}{2em}%
}
\appto\biburlsetup{\Urlmuskip=0mu plus 1mu\relax}

\newtheorem{theorem}{Theorem}[section]
\newtheorem{proposition}[theorem]{Proposition}
\newtheorem{lemma}[theorem]{Lemma}

\newtheorem{corollary}[theorem]{Corollary}
\theoremstyle{definition}
\newtheorem{definition}[theorem]{Definition}
\theoremstyle{remark}

\allowdisplaybreaks
\setlist[itemize]{leftmargin=2em,itemsep=0.25em,topsep=0.4em}
\setlist[enumerate]{leftmargin=2.2em,itemsep=0.35em,topsep=0.4em}

\newcommand{\NP}{\textup{NP}}
\newcommand{\PH}{\textup{PH}}
\newcommand{\QMA}{\textup{QMA}}
\newcommand{\QMAt}{\QMA(2)}
\newcommand{\PQSigma}{\textup{PureQ}\Sigma}
\newcommand{\QSigma}{\textup{Q}\Sigma}
\newcommand{\NEXP}{\textup{NEXP}}

\newcommand{\R}{\mathbb R}
\newcommand{\C}{\mathbb C}
\newcommand{\E}{\mathbb E}
\renewcommand{\L}{\mathcal L}
\newcommand{\complex}{\mathbb C}

\newcommand{\be}{\begin{equation}}
\newcommand{\ee}{\end{equation}}
\renewcommand{\epsilon}{\varepsilon}
\DeclareMathOperator{\Tr}{Tr}
\DeclareMathOperator{\poly}{poly}
\DeclareMathOperator{\polylog}{polylog}
\DeclarePairedDelimiter{\ket}{\lvert}{\rangle}
\DeclarePairedDelimiter{\abs}{\lvert}{\rvert}
\DeclarePairedDelimiter{\norm}{\lVert}{\rVert}
\newcommand{\proj}[1]{\mbox{$|#1\rangle \!\langle #1 |$}}

\newcommand{\Dens}{\mathsf D}
\newcommand{\Sep}{\operatorname{Sep}}
\newcommand{\disttr}{\operatorname{dist}_{\mathrm{tr}}}
\newcommand{\rankpsd}{\operatorname{rank}_{\mathrm{psd}}}
\newcommand{\Herm}{\operatorname{Herm}}
\newcommand{\one}{\mathbf 1}
\newcommand{\inner}[2]{\left\langle#1,#2\right\rangle}
\newcommand{\Bin}{\operatorname{Bin}}
\newcommand{\leanref}[3]{\href{https://github.com/DorianRudolph/disentangler-lower-bounds-lean/blob/main/#1\#L#2}{\texttt{#3}}}

\title{Semidefinite extension complexity of the separable set, with applications to approximate disentanglers}
\author[1]{Sevag Gharibian}
\author[2]{Carsten Hecht}
\author[1]{Dorian Rudolph}
\affil[1]{\small Paderborn University and PhoQS, Warburger Stra{\ss}e 100, 33098 Paderborn, Germany}
\affil[2]{Paderborn University, Warburger Stra{\ss}e 100, 33098 Paderborn, Germany}
\date{}

\begin{document}
\maketitle

\begin{abstract}
Let $\Sep(d:d)$ be the set of separable states on $\C^d\otimes\C^d$, and for a two-outcome measurement operator $Q$ let $h_{\Sep}(Q)=\max_{\sigma\in\Sep(d:d)}\Tr(Q\sigma)$ be its maximum acceptance probability over separable states. This is the optimization problem underlying $\QMAt$, the class of quantum Merlin-Arthur protocols with two unentangled proofs. We prove lower bounds on the size of semidefinite programs (SDPs) that approximate $h_{\Sep}$ in the SDP extended-formulation model of Harrow, Natarajan, and Wu (HNW). For each dimension $d$, all measurement operators $Q$ share one SDP feasible region, and each product state has a feasible representative that reproduces its acceptance probability for every $Q$; only the objective varies with $Q$. For every $0<\theta<2/7$, there are constants $c_\theta,a_\theta>0$ such that every such SDP approximating $h_{\Sep}$ to additive error $a\le a_\theta$ has size at least $d^{\,c_\theta\min\{a^{-1/3},d^\theta\}}$ for sufficiently large $d$. The same bound holds for the size of any SDP-representable convex set of states that contains $\Sep(d:d)$ and lies within trace distance $a$ of it, giving a quantitative counterpart to Fawzi's theorem that the separable set has no exact semidefinite representation. HNW proved a lower bound of $d^{\log d/\polylog\log d}$ for additive error $O(1/d^2)$; our bound applies also at sufficiently small constant error and becomes $d^{\,\Omega(d^\theta)}$ when $a\le d^{-3\theta}$. Like HNW, we use the Lee--Raghavendra--Steurer (LRS) pseudo-density lower bound. Our stronger bounds come from the quantitative part of the LRS theorem, using an explicit pseudo-density and realizing the associated matrix with high PSD rank as expectations of bounded block-positive test operators on product states.

As an application, we obtain lower bounds on the input dimension $D$ of approximate disentanglers, the channels in Watrous' disentangler conjecture. An $(\varepsilon,\delta)$-approximate disentangler is a quantum channel whose every output is within trace distance $\varepsilon$ of the separable states and whose outputs cover every separable state to within trace distance $\delta$. Such a channel yields an HNW formulation of size $O(D+d^2)$, so with $a=\varepsilon+\delta$ we obtain $D\ge d^{\,c_\theta\min\{a^{-1/3},d^\theta\}}$. At error $a=O(1/d^2)$, this improves exponentially over HNW. It also rules out the disentangler approach to proving $\QMA=\QMAt$ for superpolynomially small promise gaps $n^{-\omega(1)}$, where $n=\log_2d$. Independent and concurrent work of Bostanci et al.\ proves, via a quantum oracle separation, that for fixed $\varepsilon,\delta$ with $a<1$, one has $\log D=\Omega_{\varepsilon,\delta}(d^{\gamma(a)})$ for some $\gamma(a)>0$, with $\gamma(a)=1/2$ for $a<2/3$. This supersedes our disentangler bound; their reduction does not apply to SDP formulations. Our main results are supported by Lean proofs.
\end{abstract}

\section{Introduction}
\label{sec:introduction}

A central open question in quantum complexity theory is the power of unentangled proof systems, formally the class $\QMAt$~\cite{kobayashiQuantumMerlinArthurProof2003,gharibianGuestColumn72024,jeronimoQMA2UniverseComplexityEntanglement2026}. Here, $\QMAt$ is defined as Quantum Merlin-Arthur (QMA), except that the proof $\ket{\psi}_{AB}$ is guaranteed to be a tensor product across a prespecified cut $A$ versus $B$ of the qubits, i.e. $\ket{\psi}_{AB}=\ket{\phi}_A\otimes\ket{\varphi}_B$ for arbitrary states $\ket{\phi}_A\otimes\ket{\varphi}_B\in\complex^d\otimes \complex^d$, and where $d=2^n$ for $n$ the number of qubits. The best known upper~\cite{gharibianQuantumGeneralizationsPolynomial2022,grewalPureQuantumPolynomial2026} and lower bounds for $\QMAt$ are
\be
  \QMA\subseteq \QMAt\subseteq\PQSigma_2\subseteq\QSigma_3\subseteq\NEXP.
\ee
Here, $\PQSigma_2$ is a quantum generalization of the second level of the Polynomial Hierarchy (\PH) with \emph{pure} quantum proofs\footnote{Thus, \cite{grewalPureQuantumPolynomial2026} implies swapping the existential quantifier on $\ket{\varphi}_B$ in $\QMAt$ to a universal quantifier cannot decrease the power of the class, assuming the proof is guaranteed to be pure.}, and $\QSigma_3$ is analogous except with \emph{mixed} proofs. Thus, whether $\QMA\overset{?}{=}\QMAt$ remains a difficult open question.

\vspace{-2mm}
\paragraph{Disentanglers.} A natural approach to proving $\QMA=\QMAt$ is for a QMA verifier to wish for a quantum channel $\Lambda:\L(\complex^D)\rightarrow\L(\complex^d\otimes\complex^d)$, a \emph{disentangler}, with two properties\footnote{Formally, one requires a third property as well, that $\Lambda$ be efficiently implementable. However, as with previous works, our focus is on the stronger aim of \emph{information-theoretic} lower bounds on the input space dimension to $\Lambda$, irrespective of $\Lambda$'s actual implementation.}: (1) Given any input $\rho$, $\Lambda(\rho)$ is $\epsilon$-close to a separable state (in trace distance) across the $A$ versus $B$ cut, and (2) for any desired separable state $\sigma_{AB}$, there exists input $\rho$ such that $\Lambda(\rho)$ is $\delta$-close to $\sigma$. Given $\Lambda$, a QMA verifier can simulate $\QMAt$ by first applying $\Lambda$ to any proof $\rho$ sent by a potentially cheating prover. Then, property (1) guarantees soundness by preventing the prover from entangling $\rho$, and property (2) guarantees completeness, i.e. an honest prover can effectively send any desired separable witness $\sigma_{AB}$. More precisely, a verifier with completeness $c$ and soundness $s$ is converted into one with completeness at least $c-\delta$ and soundness at most $s+\epsilon$, so its gap is at least $(c-s)-(\epsilon+\delta)=(c-s)-a$. For brevity, we henceforth refer to this as the \emph{disentangler strategy}. Watrous conjectured (Conjecture 5.2 of~\cite{aaronsonPowerUnentanglement2009}) that for all constants $\epsilon,\delta<1$, any disentangler requires
\be\label{eqn:watrous}
  D=2^{\Omega(d)}
\ee
(i.e. the input space requires $\Omega(d)$ qubits, for $d=2^n$), implying the disentangler strategy fails. Very recently, Jeronimo, Wu, and Xu~\cite{jeronimoOptimalQuantumFinetti2026} showed that for any constant $0<\epsilon<1$ and $\delta=0$, there exists a disentangler with $\log D \in O( \sqrt{d}\log d)$, thus refuting Watrous' conjecture as originally stated. Note, however, that this does not say anything about the disentangler strategy, which we now expound upon.

Namely, while \Cref{eqn:watrous} as stated was well-motivated from an information-theoretic perspective~\cite{aaronsonPowerUnentanglement2009}, to rule out the disentangler strategy, it suffices to achieve something weaker. Writing $d=2^n$, define an \emph{obstruction} for total error $a:=\epsilon+\delta$ as any lower bound forcing
\begin{equation}\label{eqn:obstruction}
  \log_2D=n^{\omega(1)}.
\end{equation}
In words, an obstruction forces a superpolynomial number of input qubits,
which rules out using a disentangler with polynomially many input qubits to
show that $\QMA$ can simulate $\QMAt$. In this terminology, fully ruling out
the disentangler strategy for ordinary $\QMA$ requires an obstruction whenever
$1-a\ge1/\poly(n)$.\footnote{Harrow and Montanaro
showed~\cite{harrowTestingProductStates2013} that $\QMAt$ can be amplified to
completeness at least $1-2^{-p(n)}$ and soundness at most $2^{-p(n)}$ for a
polynomial $p$. Applying a disentangler with total error $a$ leaves a
single-proof gap of at least $1-a-2^{1-p(n)}$. Thus it gives an ordinary-$\QMA$
verifier whenever $1-a\ge1/\poly(n)$; if $1-a$ is only inverse exponential,
the resulting verifier may have only an inverse-exponential gap.}

\vspace{-2mm}
\paragraph{Previous work.} In the zero-error case, $a=0$, Aaronson et
al.~\cite{aaronsonPowerUnentanglement2009} confirmed Watrous' conjecture\footnote{More accurately, \cite{aaronsonPowerUnentanglement2009} shows for $a=0$ that $D=\infty$.}, obtaining an obstruction for $a=0$. Thus, the open frontier became \emph{quantitative} lower bounds on $D$ as a function of $a$. Next, Harrow, Natarajan, and Wu~\cite{harrowLimitationsSemidefinitePrograms2019} showed that for $a<1/\poly(d)$ (i.e. exponentially small error in the number of output qubits, $n$), one has $D\geq d^{\log(d)/\polylog(\log d)}$. Note this does not constitute an obstruction for $a<1/\poly(d)$. Finally, Akibue, Kato, and Tani~\cite{akibueHardnessConversionEntangled2026} showed Watrous' conjecture for the subset of \emph{strong} disentanglers\footnote{A \emph{strong} disentangler supplements our current definition with an environment $E$, so that for any $\rho$ on $ABE$, $\Lambda_{AB}\otimes I_{E}(\rho)$ is close to separable across the $A$ versus $BE$ cut. This is closely related to an approximate entanglement-breaking requirement after discarding one output subsystem. For the disentangler strategy, however, it suffices to have an entanglement-annihilating channel~\cite{moravcikovaEntanglementannihilatingEntanglementbreakingChannels2010}, i.e. one which breaks the entanglement on $AB$ after tracing out $E$.} when $\epsilon+\sqrt{\delta}<1$. This yields a partial obstruction against a \emph{subclass} of possible disentanglers. Jeronimo and Wu constructed dimension-independent disentangler-like channels from unentangled inputs~\cite{jeronimoDimensionIndependentDisentanglers2024}; their soundness guarantee assumes that the input is unentangled and therefore does not meet the present definition, which quantifies over arbitrary input states. In sum, prior to the present work and the concurrent work of Bostanci et al.~\cite{bostanciQuantumOracleSeparation2026} discussed below, no unconditional information-theoretic obstruction was known for unrestricted approximate disentanglers at $a>0$.

\vspace{-2mm}
\paragraph{Our results.}
Our main technical result concerns the maximum measurement value over separable states,
\[
 h_{\Sep(d:d)}(Q)=\max_{\sigma\in\Sep(\C^d:\C^d)}\Tr(Q\sigma),
\]
which is the optimization problem underlying $\QMAt$: the acceptance
probability of a $\QMAt$ verifier with acceptance operator $Q$, maximized
over unentangled proofs, equals $h_{\Sep}(Q)$.  It is also the target of the
symmetric-extension hierarchy of Doherty, Parrilo, and
Spedalieri~\cite{dohertyCompleteFamilySeparability2004} and of sum-of-squares
algorithms for separability
problems~\cite{jeronimoArgmaxPrincipleSum2026,andersonSimpleAlgorithmBest2026}.
We ask how large a semidefinite program must be to approximate $h_{\Sep}$,
using the SDP extended-formulation framework of Harrow, Natarajan, and Wu
(HNW)~\cite[Defs.~3.11, 5.1, and~5.2]{harrowLimitationsSemidefinitePrograms2019}.
HNW define such a formulation through an \emph{embedded reduction} from
$h_{\Sep}$ to an SDP.  Concretely, for each dimension $d$, the reduction
assigns to each measurement operator $Q$ an affine objective $\Phi_Q$ on a
common feasible region $\mathcal P$ of $r\times r$ PSD matrices.  The
embedding assigns to each pure product state $\tau$ a feasible point
$\iota(\tau)$ satisfying
\[
 \Phi_Q(\iota(\tau))=\Tr(Q\tau)\qquad\text{for every }Q.
\]
Thus the same point represents $\tau$ for all measurements.  For thresholds
$0<s<c<1$, the approximation requirement is
\[
 h_{\Sep}(Q)\le s\quad\Longrightarrow\quad
 \sup_{Y\in\mathcal P}\Phi_Q(Y)\le c.
\]
The embedding ensures that the SDP optimum is at least $h_{\Sep}(Q)$.
We give the full definition, including the convention on bounded
objectives, in \cref{def:hnw-hsep-ef}.  Let
\[
 \varrho=\frac{c-s}{\min\{s,1-s\}}.
\]
\Cref{thm:hsep-endpoint} shows that every such formulation of size
$r$ satisfies
\[
 \log r=\Omega\bigl(M_\varrho\log(2+R_\varrho)\bigr),
\]
where
\[
 M_\varrho=\min\left\{\varrho^{-1/3},
 \left(\frac{d-1}{(\log(d-1))^3}\right)^{2/7}\right\},
 \qquad
 R_\varrho=\frac{d-1}{M_\varrho^{7/2}(\log(d-1))^3}.
\]
In particular, for every fixed $0<\theta<2/7$ there are constants
$c_\theta,\eta_\theta>0$ such that, for $\varrho\le\eta_\theta$ and
sufficiently large $d$,
\be\label{eqn:intro-sdp-power}
 r\ge d^{\,c_\theta\min\{\varrho^{-1/3},d^\theta\}};
\ee
see \eqref{eq:hsep-extension-complexity}.  The same bound holds with
$\varrho=2a$ for SDPs that approximate $h_{\Sep}$ to uniform additive error
$a$ (\cref{cor:hsep-uniform-additive}).  This has a direct geometric
interpretation.  A \emph{spectrahedral shadow} of size $r$ is the image of
an SDP feasible region of $r\times r$ PSD matrices under an affine map.  Any
spectrahedral shadow that contains $\Sep(d:d)$ and lies within trace
distance $a$ of it gives a uniform additive approximation to $h_{\Sep}$, so
its size is at least $d^{\,c_\theta\min\{(2a)^{-1/3},d^\theta\}}$
(\cref{cor:sep-spectrahedral-shadow}).  This is a quantitative counterpart
to Fawzi's theorem that the separable set has no exact semidefinite
representation~\cite{fawziSetSeparableStates2021}.  The product states used
in the proof are symmetric powers $\proj\psi^{\otimes2}$, and the operator
used before adding a multiple of $I-F$ is nonnegative on every such power,
so the same lower bounds hold for the two-particle bosonic separable body
(\cref{cor:bosonic-hsep}).

\paragraph{Comparison with HNW.}
HNW proved the first lower bound in this
model~\cite[Thm.~5.6]{harrowLimitationsSemidefinitePrograms2019}: at
thresholds $(1-\epsilon(d),1-\delta(d))$ with
$0<\epsilon(d)<\delta(d)=O(1/d^2)$, they obtain size at least
$d^{\log d/\polylog\log d}$.  Both thresholds are within $O(1/d^2)$ of
one.  Our theorem is parameterized by the relative gap $\varrho$ and also
allows constant thresholds such as $(1/2+a,1/2)$ for sufficiently small
constant $a>0$.  The endpoint $c=1$ is uninformative in our bounded-objective
convention: optimizing over all density matrices already gives a
$(1,s)$-approximate formulation for every $s<1$.

For uniform additive approximations, the comparison is direct.  If
$0<a<1/2$, an SDP within additive error $a$ of $h_{\Sep}$ is a
$(1-a,1-2a)$-approximate formulation.  HNW's bound therefore applies at
sufficiently small inverse-square error in $d$.
\Cref{cor:hsep-uniform-additive} applies to every $a\le a_\theta$.
For fixed positive $a$, our bound is polynomial in $d$, with an exponent
that grows as $a^{-1/3}$ as $a$ decreases.  It is superpolynomial when
$a=o(1)$ and gives size $d^{\,c_\theta d^\theta}$ once
$a\le d^{-3\theta}$.  In particular, at the inverse-square accuracy
considered by HNW, it improves their quasipolynomial bound to one that is
exponential in a power of $d$.

\paragraph{Application to approximate disentanglers.}
An $(\epsilon,\delta)$-approximate disentangler with input dimension $D$
yields a formulation of the above kind with relative gap $\varrho=2a$ and
size $O(D+d^2)$ (\cref{sec:disentangler-reduction}), so the SDP lower bound
transfers to $D$.  The two errors enter this reduction only through their
sum.  When $a=0$, we interpret $(2a)^{-1/3}=+\infty$.

\begin{theorem}[Input-dimension lower bound for approximate disentanglers]
\label{thm:dimension}
There are universal constants $c,a_0>0$ and $d_0\in\mathbb N$ such that the
following holds.  Let
\[
 \Lambda:\Dens(\C^D)\longrightarrow
 \Dens(\C^{d_A}\otimes\C^{d_B})
\]
be an $(\varepsilon,\delta)$-approximate disentangler, put
$a=\varepsilon+\delta$ and $d=\min\{d_A,d_B\}$, and assume
$a\le a_0$ and $d\ge d_0$.  Put
\begin{equation}
 N=d-1,
 \qquad x=\log N,
 \qquad
 M_a=\min\left\{(2a)^{-1/3},\left(\frac{N}{x^3}\right)^{2/7}\right\},
 \qquad
 R_a=\frac{N}{M_a^{7/2}x^3}.
 \label{eq:disentangler-endpoint-scales}
\end{equation}
Then $R_a\ge1$ and
\begin{equation}
 \log D\ge cM_a\log(2+R_a).
 \label{eq:dimension-endpoint}
\end{equation}
\end{theorem}

\noindent\emph{Lean:}
\leanref{Disentangler/ManuscriptAudit.lean}{156}{dimension}. Our main results are supported by Lean proofs \cite{gharibianLeanProofsDisentangler2026}.

\begin{corollary}[Accuracy-dependent growth of the input dimension]
\label[corollary]{cor:dimension-power}
Fix $0<\theta<2/7$.  There exist constants
$c_\theta,a_\theta>0$ and $d_\theta\in\mathbb N$ such that, under the same
notation, every $(\varepsilon,\delta)$-approximate disentangler with
$d\ge d_\theta$ and $a\le a_\theta$ satisfies
\begin{equation}
 D\ge d^{\,c_\theta\min\{a^{-1/3},d^\theta\}}.
 \label{eq:dimension-power}
\end{equation}
\end{corollary}

\noindent\emph{Lean:}
\leanref{Disentangler/ManuscriptAudit.lean}{189}{dimension\_power}.

\noindent In words, and independently of the concurrent work of Bostanci et al.~\cite{bostanciQuantumOracleSeparation2026} discussed below, we obtain:
\begin{enumerate}
    \item To our knowledge, the first unconditional information-theoretic obstruction for unrestricted approximate disentanglers throughout the nonzero regime $a=n^{-\omega(1)}$.
    In other words, if $a$ is smaller than every inverse polynomial, then $\log_2D=n^{\omega(1)}$. Thus the disentangler strategy cannot simulate $\QMAt$ whenever the channel error is sufficiently below a promise gap in this regime.

    To put this in context, $\QMAt$ with an inverse-exponential completeness--soundness gap $\Delta$ equals $\NEXP$~\cite{pereszlenyiMultiProverQuantumMerlinArthur2012}. A disentangler-based simulation must have channel error $a$ sufficiently below $\Delta$ to preserve this gap. Since $\textup{PreciseQMA}=\textup{PSPACE}$~\cite{feffermanCompleteCharacterizationUnitary2018}, unless $\textup{PSPACE}=\NEXP$, an efficiently implementable disentangler with $a=o(\Delta)$, diamond-norm implementation error $o(\Delta)$, and $\log D=\poly(n)$ was already unlikely to exist. Our result is instead an \emph{information-theoretic} obstruction which (a) makes no assumption on the complexity of implementing $\Lambda$, and (b) extends throughout the regime $a=n^{-\omega(1)}$, in which nothing is known about the complexity of $\QMAt$.
    \item For $a\leq 1/n^c$ for constant $c$, i.e. inverse-polynomial channel error relative to the output qubit count, we obtain, to our knowledge, the first lower bounds that force a superlinear number $\log D$ of input qubits:
    \be\label{eqn:poly}
       D\ge2^{\Omega(n^{1+c/3})}.
    \ee
    This extends the admissible error regime exponentially compared to~\cite{harrowLimitationsSemidefinitePrograms2019}, which obtained similar lower bounds for $a=2^{-\poly(n)}$. As a concrete example, if $a\leq 1/n^3$, \Cref{eqn:poly} says any disentangler $\Lambda$ requires $\log D\geq \Omega(n^2)$, i.e. the input to $\Lambda$ consists of quadratically more qubits than its output.
\end{enumerate}

\noindent The two terms in $M_a$ are equal when $a$ is of order
$({\log d})^{18/7}/d^{6/7}$.  If
\[
 a=O\!\left(\frac{(\log d)^{18/7}}{d^{6/7}}\right),
\]
then
$M_a=\Theta\bigl(d^{2/7}/(\log d)^{6/7}\bigr)$, and therefore
\begin{equation}
 \log D
 =\Omega\!\left(\frac{d^{2/7}}{(\log d)^{6/7}}\right),
 \label{eq:intro-saturated}
\end{equation}
closing the fixed-power loss in \Cref{cor:dimension-power}.

\paragraph{The SDP model and related work.}
We emphasize the model: the lower bound applies to SDPs with a common
objective-independent feasible region and an objective-independent embedding
of product states; it does not by itself lower-bound arbitrary algorithms or
formulations whose constraints may be rebuilt without restriction for each
objective.
Concurrent works give SoS-based algorithms for the 
promise $h_{\Sep}(M)=1$ versus $h_{\Sep}(M)\le 1-\varepsilon$, using
constraints derived from the eigenvalue-$1$ eigenspace of $M$ and subsequent
rounding~\cite{jeronimoArgmaxPrincipleSum2026,andersonSimpleAlgorithmBest2026}.
These results are complementary to our lower bounds, which only rule out certain types of SDP formulations.

Among concrete SDP approaches to separability, Doherty, Parrilo, and
Spedalieri introduced the symmetric-extension hierarchy, whose successive
outer approximations converge to the separable set
~\cite{dohertyCompleteFamilySeparability2004}.  Two other earlier results
address related but distinct representation models.
Fawzi proved that, beyond the low-dimensional PPT cases, the separable set
has no finite exact semidefinite representation~\cite{fawziSetSeparableStates2021}.
Aubrun and Szarek gave quantitative lower bounds on the number of positive-map
tests needed to detect all robustly entangled states~\cite{aubrunDvoretzkyComplexityEntanglement2017}.
The former result is exact and nonquantitative, while the latter concerns a
restricted family of tests.  Our theorem instead gives quantitative lower
bounds for general objective-independent PSD lifts that approximate the
separable set in trace distance.

\paragraph{Parallel work.}
Independent and concurrent work by Bostanci, Grewal, Haferkamp, Huang,
Hwang, Natarajan, and Nirkhe~\cite{bostanciQuantumOracleSeparation2026}
proves a quantum-oracle separation between $\QMAt$ and $\QMA$ and stronger
lower bounds for approximate disentanglers.  In our notation, for every
fixed $\epsilon,\delta\ge0$ with $a=\epsilon+\delta<1$, their result shows
that the number $\log_2D$ of input qubits is exponential in the number
$n=\log_2d$ of qubits in each output register.  For $a<2/3$, they prove
$\log D=\Omega_{\epsilon,\delta}(\sqrt d)$, and for $2/3<a<1$ they prove
$\log D=\Omega_{\epsilon,\delta}(d^{\gamma})$ for a constant $\gamma>0$
depending on $a$.  Since an
$(\epsilon,\delta)$-approximate disentangler is also an
$(\epsilon',\delta')$-approximate disentangler for all $\epsilon'\ge\epsilon$
and $\delta'\ge\delta$, this bound supersedes \cref{thm:dimension} as a
statement about disentanglers.  Their reduction uses a
disentangler to construct a $\QMA$ verifier for a unitary-oracle problem and
then applies a query lower bound.  A general objective-independent SDP
formulation need not define either a quantum channel or a $\QMA$ verifier,
so this reduction does not establish our SDP results: the threshold lower
bound in \cref{thm:hsep-endpoint}, the uniform additive lower bound in
\cref{cor:hsep-uniform-additive}, the spectrahedral approximation lower bound
in \cref{cor:sep-spectrahedral-shadow}, or the corresponding bosonic bounds
in \cref{cor:bosonic-hsep}.  In each of these settings, our results give the
bound \eqref{eq:hsep-extension-complexity}.

\paragraph{Techniques.} We follow the approach of Harrow, Natarajan, and Wu~\cite{harrowLimitationsSemidefinitePrograms2019} (HNW). We form a nonnegative table $M$ of measurement outcomes. Its columns correspond to selected separable states $\sigma_S$, and its rows correspond to selected Hermitian operators $W_T$, called \emph{witnesses}. The entry in row $T$ and column $S$ is $\Tr[W_T\sigma_S]\ge0$. The \emph{PSD rank} of $M$ is the smallest $r$ for which every row $T$ can be assigned a positive-semidefinite (PSD) $r\times r$ matrix $A_T$ and every column $S$ a PSD $r\times r$ matrix $B_S$ such that $M(T,S)=\Tr[A_TB_S]$. A disentangler with input dimension $D$ gives a PSD factorization whose size is controlled by $D$. A theorem of Lee, Raghavendra, and Steurer (LRS)~\cite{leeLowerBoundsSize2015} shows that such a factorization requires $r\ge d^{\Omega(m)}$, where $m$ is a parameter that we choose. A larger disentangler error forces us to choose a smaller $m$. Comparing the upper and lower bounds on $r$ gives \Cref{thm:dimension}.

To explain the connection to PSD rank, recall that a Hermitian operator is PSD exactly when its expectation in every state is nonnegative. Our witnesses are instead \emph{block positive}: they have nonnegative expectation in every separable state but can have negative expectation in an entangled state. Thus the states and witnesses give the entries of $M$, but they do not themselves form a PSD factorization. Now apply $W_T$ to the output of a disentangler. On an input state $\rho$, its expectation is $\Tr[W_T\Lambda(\rho)]$. Property (1) ensures that $\Lambda(\rho)$ is within $\epsilon$ of a separable state, on which $W_T$ has nonnegative expectation. The expectation is therefore at least $-2\epsilon\norm{W_T}_\infty$ for every input $\rho$. Adding $2\epsilon\norm{W_T}_\infty$ times the identity makes the resulting input-space operator PSD. In the table, this operation adds the same small shift to every entry in a row. Property (2) gives an approximate input preimage $\rho_S$ for each column state $\sigma_S$. These input states form the other side of the factorization. Hence a disentangler gives a PSD factorization of a slightly shifted version of $M$. More generally, any SDP that approximates $h_{\Sep}$ with a common feasible region gives the same kind of factorization. We first prove a lower bound for these SDPs (\Cref{thm:hsep-endpoint}) and then construct one of size $O(D+d^2)$ from a disentangler.

The PSD-rank lower bound is an application of the LRS theorem. Let $f$ be a polynomial in $m$ Boolean variables that is nonnegative on every point of $\{0,1\}^m$. A degree-$k$ \emph{pseudo-expectation} with margin $g>0$ is a linear functional $\widetilde{\E}$ on polynomials of degree at most $k$ such that $\widetilde{\E}[1]=1$, $\widetilde{\E}[p^2]\ge0$ for every $p$ of degree at most $k/2$, and $\widetilde{\E}[f]=-g$. Thus it satisfies the basic positivity test available to low-degree polynomials even though it assigns a negative value to $f$. Form the table with entries $f(x_S)+b$, indexed by the $m$-element subsets $S\subseteq[N]$ and the points $x\in\{0,1\}^N$, where $x_S$ is the restriction of $x$ to $S$ and $b\ge0$. LRS show that if a pseudo-expectation of degree $\Omega(m)$ exists and $g>b$, then every PSD factorization of this table has size $N^{\Omega(m)}$, provided $N$ is at least a fixed power of $m$. The disentangler errors contribute to $b$, so we must keep the resulting shift below $g$.

HNW build their table from a $\QMAt$ protocol with logarithmic-size proofs, of the type used to show that $\NP$ is contained in $\QMAt$~\cite{blierQuantumCharacterizationNP2012,gallQMAProtocolsTwo2012}. In this protocol, two unentangled proofs of $O(\log n)$ qubits convince a verifier that a graph on $n$ vertices is $3$-colorable, so the local dimension is $d=\poly(n)$. The protocol has an inverse-polynomial promise gap: the honest proofs are superpositions over all $n$ vertices, and cheating is detected with probability $1/\poly(n)$. Consequently, the associated margin is $g=\Theta(1/n^2)=1/\poly(d)$, and the factorization argument applies when the error is exponentially small in the number of output qubits. The LRS method embeds instances on $m\ll n$ variables, but the proof states remain superpositions over all $n$ vertices. The margin therefore scales with $n$ rather than $m$. Together with the permitted parameter range, this gives the HNW bound $D\ge d^{\log d/\polylog\log d}$.

We instead make the margin depend on the parameter $m$, not on the dimension $d$. We take $f$ to be the central knapsack polynomial $f_m(x)=[(\sum_{i=1}^m x_i-m/2)^2-1/4]/m^2$ with $m$ odd. It is nonnegative at every point of $\{0,1\}^m$: since $\sum_i x_i$ is an integer and $m/2$ is a half-integer, the square is at least $1/4$. Grigoriev~\cite{grigorievComplexityPositivstellensatzProofs2001} constructed a pseudo-expectation for $f_m$ with degree $\Omega(m)$ and margin $g=1/(4m^2)$ (\Cref{sec:knapsack}). Thus $g$ depends only on $m$, which we are free to choose. \Cref{sec:witnesses} realizes these values as measurement statistics using pure product states $\sigma_S$, indexed by $m$-element subsets $S$ of the local basis, and block-positive witnesses $W_T$ with
\be
  \Tr[W_T\sigma_S]=f_m(\one_T|_S)+\frac1{8m^2},
  \qquad
  \norm{W_T}_\infty\in O(m),
\ee
where $\one_T|_S\in\{0,1\}^m$ records which elements of $S$ lie in $T$. The entries of this table are differences between an upper threshold and the values attained by selected states; such a table is called a \emph{slack matrix}. No proof protocol is needed. The norm bound $O(m)$ is independent of $d$. Writing $a=\epsilon+\delta$, the errors change each entry by at most $(\epsilon+\delta)\norm{W_T}_\infty=O(am)$. The change therefore depends on $m$ rather than $d$, which allows the argument to handle larger errors than HNW.

The argument requires $am=O(1/m^2)$, or equivalently $am^3=O(1)$. It also requires $m$ to be at most a small power of $d$. Taking $m=\Theta(\min(a^{-1/3},d^\theta))$ therefore yields $D\ge d^{\Omega(m)}$. To improve the allowed power of $d$, we apply a polynomial separately to every entry of $M$. Specifically, let $T_\ell$ be the degree-$\ell$ Chebyshev polynomial and set $P_\ell(u)=(1+T_\ell(2u-1))/2$. This transformation increases the negative pseudo-expectation by a factor $\ell^2$ while keeping all actual table entries in $[0,1]$. The identity $P_\ell(u)=u\,Q_\ell(u)^2$ bounds the increase in PSD rank. With $\ell$ a fixed fraction of $m$, the allowed range improves from $\theta<2/11$ to $\theta<2/7$. Taking $\ell=\Theta(m/\log d)$ gives the term $d^{2/7}/(\log d)^{6/7}$ in \Cref{thm:dimension}.

\Cref{sec:knapsack} proves the lower bounds for the matrix constructed using $f$, \Cref{sec:witnesses} constructs the states and witnesses, and \Cref{sec:hsep-extension-lower-bounds} combines them to prove the $h_{\Sep}$ theorem. \Cref{sec:disentangler-reduction} then converts a disentangler into a uniform $h_{\Sep}$ formulation of size $O(D+d^2)$ and proves \Cref{thm:dimension}. When $d_A\neq d_B$, local channels reduce both registers to subspaces of dimension $\min\{d_A,d_B\}$ without increasing the error.

\paragraph{Open questions.} For fixed $a<2/3$, the parallel lower bound of
Bostanci et al.~\cite{bostanciQuantumOracleSeparation2026} leaves only a
logarithmic gap from the upper bound $\log D=O_a(\sqrt d\log d)$ of
Jeronimo, Wu, and Xu~\cite{jeronimoOptimalQuantumFinetti2026}.  Determining
the optimal dependence as $a$ approaches one remains open. On the SDP side,
the main quantitative questions are whether the $a^{-1/3}$ dependence can
be improved and whether the exponent $2/7$ can be raised. More broadly,
can the methods used to prove stronger lower bounds for approximate
disentanglers, particularly the quantum-oracle approach of Bostanci et
al.~\cite{bostanciQuantumOracleSeparation2026}, be adapted to general SDP
extended formulations of $h_{\Sep}$? In particular, can these methods yield
size lower bounds exponential in a power of $d$ at sufficiently small
constant additive error, even for formulations that do not arise from
quantum channels?

\paragraph{Generative AI disclosure.}
Generative AI was used extensively and iteratively in the development of the
proofs in this paper, including to propose and refine key proof ideas over many
rounds of interaction.  The authors critically evaluated, corrected, and
synthesized these suggestions, developed the resulting arguments, and wrote and
edited the final exposition.  The authors take full responsibility for the
correctness and content of the paper.

\paragraph{Organization.}
\Cref{sec:preliminaries} fixes notation and records the norm and
PSD-rank facts used later.  \Cref{sec:knapsack} proves the knapsack
pattern-matrix lower bounds.  \Cref{sec:witnesses} constructs bounded
witnesses and product states.  \Cref{sec:hsep-extension-lower-bounds}
proves the semidefinite extension-complexity theorem and its consequences.  \Cref{sec:disentangler-reduction} reduces
approximate disentanglers to those formulations and proves
\cref{thm:dimension}.

\section{Preliminaries}
\label{sec:preliminaries}

All logarithms are natural unless a base is displayed.  This section fixes
our conventions and collects the standard facts used in the proof.

\subsection{States, norms, and trace distance}

For a finite-dimensional Hilbert space $H$, let $\Herm(H)$ be the real vector
space of Hermitian operators and let $\Dens(H)$ be the set of positive
semidefinite operators of trace one.  We write
$\norm{X}_1=\Tr\sqrt{X^*X}$ for the trace norm and $\norm{X}_\infty$ for the
operator norm.  The normalized trace distance is
\[
 \disttr(\rho,\sigma)=\frac12\norm{\rho-\sigma}_1,
 \qquad
 \disttr(\rho,\mathcal C)=\inf_{\tau\in\mathcal C}\disttr(\rho,\tau).
\]
Trace-norm duality gives
\begin{equation}
 \abs{\Tr(XY)}\le \norm X_\infty\norm Y_1
 \label{eq:trace-holder}
\end{equation}
whenever $X$ and $Y$ have compatible dimensions.

Every Hermitian operator $Y$ can be written as the difference of its positive
and negative spectral parts (its Jordan decomposition):
\[
 Y=Y_+-Y_-,
 \qquad
 Y_+,Y_-\succeq0,
 \qquad
 Y_+Y_-=0,
\]
with $\norm Y_1=\Tr Y_++\Tr Y_-$.  If $\Tr Y=0$, then
\begin{equation}
 \Tr Y_+=\Tr Y_-=\frac12\norm Y_1.
 \label{eq:trace-zero-jordan}
\end{equation}
A Hermitian operator $Q$ with $0\preceq Q\preceq I$ represents one outcome
of a two-outcome quantum measurement; the other outcome is represented by
$I-Q$.  For density operators, trace-distance duality can therefore be written
as
\begin{equation}
 \disttr(\rho,\sigma)
 =\max_{0\preceq Q\preceq I}\Tr\bigl(Q(\rho-\sigma)\bigr)
 =\max_{0\preceq Q\preceq I}\abs*{\Tr\bigl(Q(\rho-\sigma)\bigr)}.
 \label{eq:effect-trace-duality}
\end{equation}
Indeed, the positive spectral projector of $\rho-\sigma$ attains the first
maximum, and replacing $Q$ by $I-Q$ reverses the sign because the difference
has trace zero.

For a nonempty compact set $\mathcal C\subseteq\Dens(H)$, define the maximum
measurement value over $\mathcal C$ by
\begin{equation}
 h_{\mathcal C}(Q)=\max_{\rho\in\mathcal C}\Tr(Q\rho),
 \qquad 0\preceq Q\preceq I.
 \label{eq:maximum-measurement-value-general}
\end{equation}
If $\mathcal K\subseteq\Dens(H)$ satisfies
$\sup_{\rho\in\mathcal K}\disttr(\rho,\mathcal C)\le a$, then
\begin{equation}
 h_{\mathcal K}(Q)\le h_{\mathcal C}(Q)+a
 \qquad\text{for every $Q$ with $0\preceq Q\preceq I$.}
 \label{eq:measurement-value-distance}
\end{equation}
This follows by choosing, for each $\rho\in\mathcal K$, a point
$\sigma\in\mathcal C$ within trace distance $a$ and applying
\eqref{eq:effect-trace-duality}.

Trace distance cannot increase when the same quantum channel is applied to
both states.  Consequently, locally compressing the output registers of a
disentangler cannot increase its approximation errors.  We use this fact in
\cref{lem:balanced-retraction} to reduce unequal local dimensions
$d_A,d_B$ to dimension $d=\min\{d_A,d_B\}$.

\subsection{Approximate disentanglers and spectrahedral shadows}

Following Harrow--Natarajan--Wu~\cite[Def.~5.8]{harrowLimitationsSemidefinitePrograms2019},
a completely positive trace-preserving map
\[
 \Lambda:\Dens(\C^D)\longrightarrow
 \Dens(\C^{d_A}\otimes\C^{d_B})
\]
is an $(\varepsilon,\delta)$-approximate disentangler if
\begin{align}
 \sup_{\rho\in\Dens(\C^D)}
 \disttr\bigl(\Lambda(\rho),\Sep(d_A:d_B)\bigr)&\le\varepsilon,
 \label{eq:disentangler-outer}\\
 \sup_{\sigma\in\Sep(d_A:d_B)}\inf_{\rho\in\Dens(\C^D)}
 \disttr\bigl(\Lambda(\rho),\sigma\bigr)&\le\delta.
 \label{eq:disentangler-inner}
\end{align}
Throughout we put
\begin{equation}
 a=\varepsilon+\delta,
 \qquad
 d=\min\{d_A,d_B\}.
 \label{eq:a-and-d}
\end{equation}
The first condition is the outer, or soundness, approximation; the second is
the covering, or completeness, approximation.  All state spaces and channel
images here are compact, so the displayed infima are attained.  In
particular, the two inequalities force $\varepsilon,\delta\ge0$ and justify
the closest-state and closest-input choices used below.

A \emph{spectrahedron} is the feasible region of an SDP: the set of PSD
matrices that satisfy a collection of affine equations.  We write
$\mathbb S_+^r$ for the set of real symmetric PSD $r\times r$ matrices.  A convex set
$\mathcal K$ has a real \emph{PSD lift} of size $r$ if there are a
spectrahedron $\mathcal P=\mathcal A\cap\mathbb S_+^r$, an affine map $\pi$,
and
\[
 \mathcal K=\pi(\mathcal P).
\]
In other words, $\mathcal K$ is obtained by applying an affine map, which may
discard auxiliary variables, to an SDP feasible region.  It is then called a
\emph{spectrahedral shadow} of size $r$~\cite{gouveiaLiftsConvexSets2013}.  We require every approximating
shadow of quantum states to consist of trace-one PSD matrices.  Complex Hermitian lifts may be
realified at a factor-two cost by the map used in
\cref{lem:realification}.

\subsection{Separable states, block positivity, and the swap}

A bipartite state on $A\otimes B$ is separable if it can be written as a
convex combination of product states.  Equivalently,
\[
 \Sep(A:B)
 =\operatorname{conv}\left\{
  \proj u\otimes\proj v:
  \norm u=\norm v=1
 \right\}.
\]
A Hermitian operator $W$ on $A\otimes B$ is \emph{block-positive} if
\begin{equation}
 \inner{u\otimes v}{W(u\otimes v)}\ge0
 \label{eq:block-positive-definition}
\end{equation}
for all $u\in A$ and $v\in B$.  By convexity, a block-positive operator has
nonnegative expectation on every separable state:
\begin{equation}
 \Tr(W\sigma)\ge0
 \qquad
 \text{for every }\sigma\in\Sep(A:B).
 \label{eq:block-positive-separable}
\end{equation}
Block positivity is weaker than positive semidefiniteness.  For brevity, we
call the block-positive operators constructed below witnesses, although we do
not require them to be non-PSD; thus, some $W_T$ may not be entanglement
witnesses in the standard sense.

When the two tensor factors are copies of the same space $H$, the swap
operator $F$ is defined by
\[
 F(u\otimes v)=v\otimes u.
\]
It is Hermitian and unitary, and its $+1$ and $-1$ eigenspaces are the
symmetric and antisymmetric subspaces.  Thus $I-F\succeq0$ and
\begin{equation}
 \norm{I-F}_\infty=2.
 \label{eq:swap-norm}
\end{equation}

\subsection{Positive-semidefinite rank}

Let $M:I\times J\to\R_+$ be a nonnegative matrix.  Its real PSD rank is the
least $r$ for which there are real symmetric positive semidefinite matrices
$A_i,B_j\in\R^{r\times r}$ satisfying
\[
 M(i,j)=\Tr(A_iB_j).
\]
The complex PSD rank is defined in the same way using Hermitian positive
semidefinite factors.  We write these quantities as
$\rankpsd^{\R}(M)$ and $\rankpsd^{\C}(M)$.

We repeatedly use three elementary operations on PSD factorizations.  If
$M,N$ are nonnegative matrices of the same shape, then
\begin{align}
 \rankpsd(M+N)
 &\le \rankpsd(M)+\rankpsd(N),
 \label{eq:psdrank-sum}\\
 \rankpsd(M\odot N)
 &\le \rankpsd(M)\rankpsd(N),
 \label{eq:psdrank-hadamard}
\end{align}
where $\odot$ denotes entrywise, or Hadamard, product.  The first inequality
follows by taking direct sums of the factors, and the second by taking tensor
products.  Multiplication of a matrix by a positive scalar does not change
its PSD rank.  The all-ones matrix $\mathbf J$ has PSD rank one.

A related observation will be useful in the amplification argument.  If a
real matrix $C$ has ordinary rank at most $q$, choose a rank factorization
$C(i,j)=\inner{u_i}{v_j}$ with $u_i,v_j\in\R^q$.  Then
\[
 C(i,j)^2
 =\Tr\bigl((u_iu_i^{\mathsf T})(v_jv_j^{\mathsf T})\bigr),
\]
so the entrywise square satisfies
\begin{equation}
 \rankpsd^{\R}(C^{\odot2})\le q.
 \label{eq:squared-rank-factorization}
\end{equation}
For these standard constructions, see also
\cite[Thm.~2.9(iii) and proof of Thm.~2.9(v)]{fawziPositiveSemidefiniteRank2015};
the latter argument applies analogously to arbitrary real $C$.

\begin{lemma}[Realification]
\label[lemma]{lem:realification}
For every nonnegative matrix $M$,
\[
 \rankpsd^{\R}(M)\le2\rankpsd^{\C}(M).
\]
\end{lemma}

\begin{proof}
This standard complex to real conversion is also described in
\cite[Sec.~2.2]{fawziPositiveSemidefiniteRank2015}.
For a Hermitian matrix $X$, define its realification by
\[
 \mathcal R(X)=
 \begin{pmatrix}
  \operatorname{Re}X&-\operatorname{Im}X\\
  \operatorname{Im}X&\operatorname{Re}X
 \end{pmatrix}.
\]
If $X\succeq0$, then $\mathcal R(X)\succeq0$.  One way to see this is to
write a real vector as $(u,v)$, identify it with the complex vector
$u+iv$, and check that the corresponding quadratic forms agree.
Furthermore, for Hermitian $X,Y$,
\begin{equation}
 \Tr\bigl(\mathcal R(X)\mathcal R(Y)\bigr)
 =2\Tr(XY).
 \label{eq:realification-trace}
\end{equation}
Thus a complex factorization $M(i,j)=\Tr(A_iB_j)$ of size $r$ gives a real
factorization of size $2r$ with factors
$\mathcal R(A_i)$ and $\mathcal R(B_j)/2$.
\end{proof}

\subsection{Boolean multilinear polynomials and binomial moments}
\label{subsec:boolean-background}

Every function on the Boolean cube $\{0,1\}^m$ has a unique multilinear
polynomial representation.  We therefore identify polynomials that agree on
the Boolean cube.  Algebraically, this means that we work in the quotient
\begin{equation}
 \R[X_1,\ldots,X_m]
 \big/
 \langle X_i^2-X_i:i\in[m]\rangle.
 \label{eq:boolean-quotient}
\end{equation}
Thus, whenever two polynomials are multiplied, we replace every power
$X_i^k$ with $X_i$ before applying a linear functional.  For
$S\subseteq[m]$, write
\[
 X^S=\prod_{i\in S}X_i,
 \qquad
 H=\sum_{i=1}^mX_i.
\]
On a Boolean input $x$, the value $H(x)=\abs x$ is its Hamming weight.  A
\emph{Hamming-weight layer} is the set of strings with a fixed Hamming
weight.

If $x$ is uniform on $\{0,1\}^m$, then $H(x)$ has the binomial distribution
$\Bin(m,1/2)$:
\begin{equation}
 \Pr[H=w]=2^{-m}\binom mw,
 \qquad 0\le w\le m.
 \label{eq:binomial-layer-probability}
\end{equation}
In particular,
\begin{equation}
 \E H=\frac m2,
 \qquad
 \E\left(H-\frac m2\right)^2
 =\operatorname{Var}(H)=\frac m4.
 \label{eq:binomial-first-two-moments}
\end{equation}
For any function $h$ that depends only on Hamming weight,
\begin{equation}
 \E_x h(\abs x)
 =2^{-m}\sum_{w=0}^m\binom mw h(w).
 \label{eq:layer-average}
\end{equation}

For $j\ge0$, let
\[
 (H)_j=H(H-1)\cdots(H-j+1)
\]
be the falling factorial, with $(H)_0=1$.  In the Boolean quotient,
\begin{equation}
 (H)_j=j!\sum_{\substack{S\subseteq[m]\\\abs S=j}}X^S.
 \label{eq:falling-factorial-boolean}
\end{equation}
Indeed, at a Boolean point of Hamming weight $w$, both sides equal
$(w)_j=j!\binom wj$.  The polynomials
$1,(H)_1,\ldots,(H)_m$ form a basis for the univariate polynomials in $H$ of
degree at most $m$.

\subsection{Generalized binomial coefficients and Lagrange interpolation}
\label{subsec:interpolation-background}

For a real or complex number $t$ and an integer $j\ge0$, the generalized
binomial coefficient is
\begin{equation}
 \binom tj=\frac{(t)_j}{j!}
 =\frac{t(t-1)\cdots(t-j+1)}{j!}.
 \label{eq:generalized-binomial}
\end{equation}
This agrees with the usual binomial coefficient when $t$ is a nonnegative
integer, but it is also meaningful when $t$ is a half-integer.

For the interpolation nodes $0,1,\ldots,m$, define
\begin{equation}
 c_w(s)=
 \prod_{\substack{0\le j\le m\\j\ne w}}
 \frac{s-j}{w-j},
 \qquad 0\le w\le m.
 \label{eq:lagrange-basis-prelim}
\end{equation}
The polynomial $c_w$ equals one at $s=w$ and zero at every other node.  Hence
Lagrange interpolation says that every univariate polynomial $p$ of degree at
most $m$ satisfies
\begin{equation}
 p(t)=\sum_{w=0}^m c_w(t)p(w).
 \label{eq:lagrange-interpolation}
\end{equation}
We will apply this identity at the nonintegral point $t=m/2$, especially to
the polynomials $p(w)=\binom wj$.

\subsection{Chebyshev polynomials}
\label{subsec:chebyshev-background}

The Chebyshev polynomials of the first and second kind are characterized by
\begin{align}
 T_j(\cos\varphi)&=\cos(j\varphi),
 \label{eq:chebyshev-first}\\
 U_j(\cos\varphi)&=
 \frac{\sin((j+1)\varphi)}{\sin\varphi}.
 \label{eq:chebyshev-second}
\end{align}
The endpoint values in \eqref{eq:chebyshev-second} are understood by
continuity, and we set $U_{-1}=0$.  These identities show that $T_j$ has
degree $j$, $U_j$ has degree $j$, and
\begin{equation}
 -1\le T_j(x)\le1
 \qquad\text{for }x\in[-1,1].
 \label{eq:chebyshev-bounded}
\end{equation}
For $y\ge0$, the hyperbolic counterpart is
\begin{equation}
 T_j(\cosh y)=\cosh(jy).
 \label{eq:chebyshev-hyperbolic}
\end{equation}
When $j$ is odd, $T_j$ is odd.  These elementary facts will be used to
amplify a small negative value lying just outside the interval $[0,1]$.

\section{Knapsack pattern matrices and Grigoriev's pseudo-density}
\label{sec:knapsack}

We now introduce the nonnegative matrices whose PSD rank will be compared
with the size of the semidefinite formulations in
\Cref{sec:hsep-extension-lower-bounds}.  The construction is based on the
central knapsack constraint
\[
 \sum_{i=1}^m x_i=\frac m2.
\]
For odd $m$, this equation has no Boolean solution because its right-hand side
is a half-integer.  Grigoriev's pseudo-expectation nevertheless behaves, on
low-degree polynomials, as though it were supported on this nonexistent
Hamming-weight layer. The LRS theorem uses this discrepancy to rule out small PSD factorizations.

All polynomial calculations in this section take place in the Boolean
quotient \eqref{eq:boolean-quotient}.  In particular, a square $g^2$ is
multilinearized before a linear functional is applied.

\begin{definition}[Pseudo-density]
Let $r\in\mathbb N$.  A function
$\mathcal D:\{0,1\}^m\to\R$ is a degree-$r$ pseudo-density if
\[
 \E\mathcal D=1
\]
and
\[
 \E\mathcal D(x)g(x)^2\ge0
\]
for every real polynomial $g$ of degree at most $\lfloor r/2\rfloor$.
\end{definition}

A genuine probability density would be
nonnegative pointwise.  A pseudo-density may take negative values, but it
still assigns nonnegative expectation to all low-degree squares.  This weaker
positivity is exactly what is needed by the Lee--Raghavendra--Steurer theorem.
For such a function, we write
$\norm{\mathcal D}_\infty=\max_x\abs{\mathcal D(x)}$.

For odd $m\ge3$, define the central knapsack polynomial
\begin{equation}
 f_m(z)=\frac1{m^2}
 \left[
  \left(\sum_{i=1}^mz_i-\frac m2\right)^2-\frac14
 \right].
 \label{eq:f-def}
\end{equation}
If $z\in\{0,1\}^m$, then $\sum_i z_i$ is an integer whereas $m/2$ is a
half-integer.  Hence
\[
 \abs*{\sum_i z_i-m/2}\ge\frac12,
\]
which proves $f_m(z)\ge0$.  On the formal layer
$\sum_i z_i=m/2$, however, the value would be
\begin{equation}
 f_m=-\frac1{4m^2}.
 \label{eq:formal-knapsack-value}
\end{equation}
This discrepancy between true Boolean points and the formal central layer is
the gap exploited in the proof.

For $N>m$ and $b\ge0$, define the shifted pattern matrix
\begin{equation}
 M_{N,m,b}(S,x)=f_m(x_S)+b,
 \qquad
 S\in\binom{[N]}m,
 \quad x\in\{0,1\}^N.
 \label{eq:pattern-matrix}
\end{equation}
Here $x_S$ is the restriction of $x$ to the coordinates in $S$.  Since
$f_m$ is symmetric, the ordering chosen for those coordinates is irrelevant.
The rows select $m$ coordinates, and the columns range over all Boolean
assignments to the ambient $N$ coordinates.

We use the following quantitative form of the Lee--Raghavendra--Steurer
theorem~\cite[Thm.~3.1]{leeLowerBoundsSize2015}.

\begin{theorem}[Lee--Raghavendra--Steurer]
\label{thm:lrs}
There exists a universal constant $c_{\mathrm{LRS}}>0$ such that the
following holds.  Let $m,r,N\in\mathbb N$ with $m,r>1$ and
$N>2m$, let $\eta\in(0,1]$, and let
$\mathcal D:\{0,1\}^m\to\R$ be a degree-$r$ pseudo-density.  For a matrix
$Z:\binom{[N]}m\times\{0,1\}^N\to\R$, set
\begin{align*}
 \norm Z_\infty&=\max_{S,x}\abs{Z(S,x)},\\
 \norm Z_1&=\E_{S,x}\abs{Z(S,x)},\\
 L_{\mathcal D}(Z)&=\E_{S,x}\mathcal D(x_S)Z(S,x).
\end{align*}
Suppose $Z\ge0$, $\norm Z_\infty\le1$, and
\begin{equation}
 \frac{\rankpsd^{\R}(Z)^2}{\norm Z_1}
 \le
 \left(
  \frac{c_{\mathrm{LRS}}\eta}
       {r m^2\norm{\mathcal D}_\infty}
 \right)^{r/2}
 \left(
  \frac{\eta}{\norm{\mathcal D}_\infty}
 \right)^3
 \left(\frac N{\log N}\right)^{r/2}.
 \label{eq:lrs-hypothesis}
\end{equation}
Then $L_{\mathcal D}(Z)\ge-\eta$.
\end{theorem}

We always use this theorem in contrapositive form: if
$L_{\mathcal D}(Z)<-\eta$, then the hypothesis
\eqref{eq:lrs-hypothesis} must fail, and therefore
\begin{equation}
 \rankpsd^{\R}(Z)
 >\sqrt{\norm Z_1}
 \left(
  \frac{c_{\mathrm{LRS}}\eta N}
       {r m^2\norm{\mathcal D}_\infty\log N}
 \right)^{r/4}
 \left(
  \frac{\eta}{\norm{\mathcal D}_\infty}
 \right)^{3/2}.
 \label{eq:lrs-contrapositive}
\end{equation}
The rest of this section constructs a pseudo-density for which the functional
on the left is explicitly negative.

\subsection{The Grigoriev functional and LRS pseudo-density}

Fix odd $m\ge3$ and put
\[
 t=\frac m2.
\]
For a multilinear monomial $X^S=\prod_{i\in S}X_i$, define the linear
functional
\begin{equation}
 \mathcal G_m(X^S)=
 \frac{\binom t{\abs S}}{\binom m{\abs S}},
 \label{eq:grigoriev-moments}
\end{equation}
where the numerator is the generalized binomial coefficient from
\eqref{eq:generalized-binomial}.  For every univariate polynomial $p$ of
degree at most $m$, this functional satisfies
\begin{equation}
 \mathcal G_m(p(H))=p(m/2).
 \label{eq:grigoriev-formal-evaluation-preview}
\end{equation}
Thus $\mathcal G_m$ acts on symmetric polynomials as evaluation on the
nonexistent Boolean layer of Hamming weight $m/2$.  We prove this identity
below in \eqref{eq:formal-layer-evaluation}.

\begin{theorem}[Grigoriev positivity]
\label{thm:grigoriev}
For every real multilinear polynomial $p$ of degree at most
$\lfloor m/2\rfloor$,
\[
 \mathcal G_m(p^2)\ge0.
\]
\end{theorem}

\noindent\emph{Lean:}
\leanref{Disentangler/ManuscriptAudit.lean}{97}{grigoriev\_positivity}.

This is the positivity statement in Grigoriev's knapsack lower
bound~\cite[Lem.~1.4]{grigorievComplexityPositivstellensatzProofs2001}.
Lee, Raghavendra, and Steurer represented $\mathcal G_m$ by the following
explicit density on the Boolean cube using Lagrange
interpolation~\cite[Sec.~5.1 and Thm.~5.3]{leeLowerBoundsSize2015}.  They bounded the
infinity norm of this density by $O(m^{3/2})$.  Lee, Prakash, de Wolf, and
Yuen later pointed out that correcting this calculation improves the bound from order $m^{3/2}$ to
$m^{1/2}$~\cite[Sec.~1.3.1, n.~3]{leeSumSquaresDegree2016}. We make this calculation explicit 
down below.

For $w\in\{0,\ldots,m\}$, let $c_w$ be the Lagrange polynomial from
\eqref{eq:lagrange-basis-prelim}, and define
\begin{equation}
 \mathcal D_m(x)=
 2^m\frac{c_{\abs x}(m/2)}{\binom m{\abs x}}.
 \label{eq:explicit-density}
\end{equation}
The function is constant on every Hamming-weight layer.  It is generally not
pointwise nonnegative.

\begin{lemma}[LRS pseudo-density with an asymptotically tight norm bound]
\label[lemma]{lem:explicit-pseudodensity}
The function $\mathcal D_m$ is a degree-$m$ pseudo-density and satisfies
\begin{align}
 \norm{\mathcal D_m}_\infty&\le\sqrt m,
 \label{eq:explicit-norm}\\
 \E_x\mathcal D_m(x)
 \left(\abs x-\frac m2\right)^2&=0.
 \label{eq:explicit-annihilate}
\end{align}
Consequently,
\[
 \E_x\mathcal D_m(x)f_m(x)=-\frac1{4m^2}.
\]
\end{lemma}

\noindent\emph{Lean:}
\leanref{Disentangler/ManuscriptAudit.lean}{106}{explicit\_pseudodensity}.

\begin{proof}
For completeness, we first verify the LRS representation identity.  We then
give the central-layer calculation and the sharper norm estimate.

\paragraph{Representation of $\mathcal G_m$.}
Let $S\subseteq[m]$ and write $\ell=\abs S$.  The function
$\mathcal D_m(x)$ depends only on $\abs x$, so by permutation symmetry all
monomials of degree $\ell$ have the same weighted expectation.  Since
\[
 \sum_{\substack{R\subseteq[m]\\\abs R=\ell}}x^R
 =\binom{\abs x}{\ell},
\]
we obtain
\begin{equation}
 \E_x\mathcal D_m(x)x^S
 = \E_x \mathcal D_m(x) \frac1{\binom m\ell}\sum_{\substack{R\subseteq[m]\\\abs R=\ell}}x^R
 =\frac1{\binom m\ell}
   \E_x\mathcal D_m(x)\binom{\abs x}{\ell}.
 \label{eq:monomial-to-layer-average}
\end{equation}
Now group the expectation by Hamming-weight layers.  Using
\eqref{eq:explicit-density}, the factor $2^m$ cancels the uniform probability
$2^{-m}$ and the factor $\binom mw$ cancels the number of strings in the
layer $\abs x=w$.  Therefore
\begin{align}
 \E_x\mathcal D_m(x)\binom{\abs x}{\ell}
 &=2^{-m}\sum_{w=0}^m
   \binom mw
   \left(2^m\frac{c_w(t)}{\binom mw}\right)
   \binom w\ell
 \notag\\
 &=\sum_{w=0}^m c_w(t)\binom w\ell
 =\binom t\ell.
 \label{eq:lagrange-layer-computation}
\end{align}
The last equality is Lagrange interpolation
\eqref{eq:lagrange-interpolation}, applied to the degree-$\ell$ polynomial
$w\mapsto\binom w\ell$.  Combining
\eqref{eq:monomial-to-layer-average} and
\eqref{eq:lagrange-layer-computation} gives
\[
 \E_x\mathcal D_m(x)x^S
 =\frac{\binom t\ell}{\binom m\ell}
 =\mathcal G_m(X^S).
\]
By linearity,
\begin{equation}
 \E_x\mathcal D_m(x)p(x)=\mathcal G_m(p)
 \label{eq:density-represents-G}
\end{equation}
for every multilinear polynomial $p$.

Taking $p=1$ in \eqref{eq:density-represents-G} gives
$\E\mathcal D_m=1$.  If $g$ has degree at most $\lfloor m/2\rfloor$, then
\Cref{thm:grigoriev} and
\eqref{eq:density-represents-G} give
\[
 \E_x\mathcal D_m(x)g(x)^2
 =\mathcal G_m(g^2)\ge0.
\]
Thus $\mathcal D_m$ is a degree-$m$ pseudo-density.

\paragraph{Evaluation on the formal central layer.}
Let $H=\sum_iX_i$.  From \eqref{eq:grigoriev-moments},
\[
 \mathcal G_m(H)
 =\sum_{i=1}^m\mathcal G_m(X_i)
 =m\frac tm=t.
\]
In the Boolean quotient,
\[
 H^2=H+2\sum_{1\le i<j\le m}X_iX_j.
\]
Consequently,
\begin{align*}
 \mathcal G_m(H^2)
 &=t+2\binom m2
   \frac{\binom t2}{\binom m2}\\
 &=t+2\binom t2
 =t+t(t-1)
 =t^2.
\end{align*}
It follows that
\[
 \mathcal G_m((H-t)^2)
 =\mathcal G_m(H^2)-2t\mathcal G_m(H)+t^2
 =0.
\]
Using \eqref{eq:density-represents-G} proves
\eqref{eq:explicit-annihilate}.  Substituting the definition
\eqref{eq:f-def} then gives
\[
 \E_x\mathcal D_m(x)f_m(x)
 =\frac1{m^2}\left(0-\frac14\right)
 =-\frac1{4m^2}.
\]

\paragraph{Infinity-norm estimate.}
Write $m=2r+1$ and let $w=\abs x$.  Starting from the definition of the
Lagrange basis,
\[
 c_w(t)
 =\frac{\prod_{j=0}^m(t-j)}
        {(t-w)\prod_{j\ne w}(w-j)}.
\]
Since
\[
 \abs*{\prod_{j\ne w}(w-j)}=w!(m-w)!,
 \qquad
 \binom mw=\frac{m!}{w!(m-w)!},
\]
we obtain
\begin{equation}
 \abs{\mathcal D_m(x)}
 =\frac{2^m}{m!}
   \frac{\abs*{\prod_{j=0}^m(t-j)}}{\abs{t-w}}.
 \label{eq:density-norm-product}
\end{equation}
The factors $t-j$ are the half-integers
$r+\tfrac12,r-\tfrac12,\ldots,-r-\tfrac12$, so
\[
 \abs*{\prod_{j=0}^m(t-j)}
 =\left(\prod_{k=0}^r\left(k+\frac12\right)\right)^2
 =\left(\prod_{k=0}^r\left(\frac{2k+1}{2}\right)\right)^2
 =\frac{\left(\prod_{k=0}^r\left({2k+1}\right)\right)^2}{2^{m+1}}.
\]
Also $\abs{t-w}=\abs{m-2w}/2$.  Substituting these two identities into
\eqref{eq:density-norm-product} yields
\begin{align}
 \abs{\mathcal D_m(x)}
 &=\frac{\left(\prod_{k=0}^r\left({2k+1}\right)\right)^2}{m!\abs{m-2w}}
 \notag\\
 &=m\left(\prod_{j=1}^r\frac{2j-1}{2j}\right)
   \frac1{\abs{m-2w}}.
 \label{eq:density-explicit-absolute}
\end{align}
Because $m$ is odd, $\abs{m-2w}\ge1$.  It remains to bound the middle factor.  Set
\[
 a_r=\prod_{j=1}^r\frac{2j-1}{2j}.
\]
We claim $a_r\le(2r+1)^{-1/2}$.  The claim is true at $r=0$.  If it holds at
$r$, then
\[
 a_{r+1}
 =a_r\frac{2r+1}{2r+2}
 \le\frac1{\sqrt{2r+1}}\frac{2r+1}{2r+2}
 \le\frac1{\sqrt{2r+3}},
\]
where the last inequality is equivalent, after squaring, to
$(2r+1)(2r+3)\le(2r+2)^2$.  Hence
\[
 \abs{\mathcal D_m(x)}
 \le m\cdot\frac1{\sqrt m}\cdot1
 =\sqrt m.
\]
This proves \eqref{eq:explicit-norm}.  The asymptotic order is tight for this
pseudo-density.  When $w=r$ or $w=r+1$,
$\abs{m-2w}=1$, and
\[
 a_r=\prod_{j=1}^r\frac{2j-1}{2j}=\Theta(m^{-1/2}).
\]
\Cref{eq:density-explicit-absolute} therefore gives
$\norm{\mathcal D_m}_\infty=\Theta(\sqrt m)$.
\end{proof}

\begin{proposition}[Explicit shifted pattern-matrix bound]
\label[proposition]{prop:explicit-rank}
There are universal constants $c_1,c_2>0$ such that, for odd $m\ge3$,
$N>2m$, and $0\le b\le3/(16m^2)$,
\begin{equation}
 \rankpsd^{\R}(M_{N,m,b})
 \ge
 c_1m^{-17/4}
 \left(
  \frac{c_2N}{m^{11/2}\log N}
 \right)^{m/4}.
 \label{eq:explicit-rank}
\end{equation}
\end{proposition}

\begin{proof}
Let $\mathcal D=\mathcal D_m$.  For each fixed row $S$, the restriction
$x_S$ of a uniform $x\in\{0,1\}^N$ is uniform on $\{0,1\}^m$.  Therefore
\cref{lem:explicit-pseudodensity} gives
\begin{align}
 L_{\mathcal D}(M_{N,m,b})
 &=\E_y\mathcal D(y)\bigl(f_m(y)+b\bigr)
 \notag\\
 &=-\frac1{4m^2}+b
 \le-\frac1{16m^2}.
 \label{eq:explicit-functional-gap}
\end{align}

We next verify the normalization required by \cref{thm:lrs}.  The
maximum of $f_m$ on the Boolean cube occurs at Hamming weight $0$ or $m$, so
\[
 0\le f_m\le\frac{m^2-1}{4m^2}.
\]
Together with $b\le3/(16m^2)$, this implies
$0\le M_{N,m,b}<1$.  Since the matrix is nonnegative,
$\norm{M_{N,m,b}}_1$ is simply its average entry.  If
$H\sim\Bin(m,1/2)$, then by
\eqref{eq:binomial-first-two-moments},
\begin{align}
 \E f_m
 &=\frac1{m^2}
   \left[
    \E\left(H-\frac m2\right)^2-\frac14
   \right]
 \notag\\
 &=\frac1{m^2}\left(\frac m4-\frac14\right)
 =\frac{m-1}{4m^2}.
 \label{eq:average-fm}
\end{align}
Hence, for $m\ge3$,
\begin{equation}
 \norm{M_{N,m,b}}_1
 =\frac{m-1}{4m^2}+b
 \ge\frac1{8m}.
 \label{eq:pattern-l1-lower}
\end{equation}

Apply the contrapositive form \eqref{eq:lrs-contrapositive} with pseudo-degree
$r=m$ and
\[
 \eta=\frac1{32m^2}.
\]
The functional value in \eqref{eq:explicit-functional-gap} is strictly less
than $-\eta$, so
\begin{equation}
 \rankpsd^{\R}(M_{N,m,b})
 >
 \sqrt{\norm{M_{N,m,b}}_1}
 \left(
  \frac{c_{\mathrm{LRS}}\eta N}
       {m^3\norm{\mathcal D}_\infty\log N}
 \right)^{m/4}
 \left(
  \frac{\eta}{\norm{\mathcal D}_\infty}
 \right)^{3/2}.
 \label{eq:explicit-lrs-substitution}
\end{equation}
Using $\norm{\mathcal D}_\infty\le\sqrt m$, the expression inside the
$m/4$ power is at least a universal constant times
\[
 \frac{N}{m^{2+3+1/2}\log N}
 =\frac{N}{m^{11/2}\log N}.
\]
The factors outside that power satisfy, up to universal constants,
\[
 \sqrt{\norm{M_{N,m,b}}_1}
 \left(\frac{\eta}{\norm{\mathcal D}_\infty}\right)^{3/2}
 \ge
 m^{-1/2}\left(m^{-2}m^{-1/2}\right)^{3/2}
 =m^{-17/4}.
\]
Absorbing numerical constants into $c_1$ and $c_2$ proves
\eqref{eq:explicit-rank}.
\end{proof}

\subsection{Central-knapsack Chebyshev--Hadamard amplification}

The preceding proposition uses the formal value
$-1/(4m^2)$ only once.  We now apply a polynomial that is bounded on the true
entry range $[0,1]$ but grows quadratically in its degree at a nearby negative
point.  A special square factorization of that polynomial keeps the resulting
PSD-rank cost under control.

For odd $\ell=2s+1$, define
\begin{equation}
 P_\ell(u)=\frac{1+T_\ell(2u-1)}2,
 \qquad
 Q_\ell(u)=U_s(2u-1)-U_{s-1}(2u-1),
 \label{eq:chebyshev-def}
\end{equation}
where $U_{-1}=0$.  The key identity is
\begin{equation}
 P_\ell(u)=uQ_\ell(u)^2.
 \label{eq:chebyshev-square}
\end{equation}
To verify it, first let $2u-1=\cos\varphi$.  Then
$u=\cos^2(\varphi/2)$, and the sine subtraction formula gives
\begin{align*}
 Q_\ell(u)
 &=\frac{\sin((s+1)\varphi)-\sin(s\varphi)}{\sin\varphi}\\
 &=\frac{\cos(\ell\varphi/2)}{\cos(\varphi/2)}.
\end{align*}
Consequently,
\[
 uQ_\ell(u)^2
 =\cos^2(\ell\varphi/2)
 =\frac{1+\cos(\ell\varphi)}2
 =P_\ell(u).
\]
The calculation holds away from the removable endpoint singularities and
therefore, by polynomial identity, holds for every $u$.  It also shows that
$P_\ell$ maps $[0,1]$ into $[0,1]$ and that
$\deg Q_\ell=(\ell-1)/2$.

For a matrix $B$, the notation $P_\ell[B]$ means entrywise polynomial
evaluation, not ordinary matrix functional calculus.  Let $\mathbf J$ denote
the all-ones matrix of the same shape as $M_{N,m,b}$, and put
\begin{equation}
 g_{m,\ell}=\frac{\ell^2}{16m^2},
 \qquad
 Z_{N,m,b,\ell}
 =\frac{P_\ell[M_{N,m,b}]+(g_{m,\ell}/2)\mathbf J}
        {1+g_{m,\ell}/2}.
 \label{eq:amplified-matrix}
\end{equation}
The small positive constant added in the numerator ensures that the average
entry is not too small, while retaining a negative pseudo-expectation.

\begin{proposition}[Amplified shifted pattern-matrix bound]
\label[proposition]{prop:amplified-rank}
Let $m\ge3$ and $\ell\ge1$ be odd, let $2\ell\le m$, let $N>2m$, and let
$0\le b\le3/(16m^2)$.  Then $0\le Z_{N,m,b,\ell}\le1$ and
\begin{equation}
 \rankpsd^{\R}(Z_{N,m,b,\ell})
 >
 \frac1{4096}\ell^4m^{-19/4}
 \left(
  \frac{c_{\mathrm{LRS}}\ell^2N}
       {64m^{11/2}\log N}
 \right)^{m/4}.
 \label{eq:amplified-rank-lower}
\end{equation}
Moreover,
\begin{equation}
 \rankpsd^{\R}(Z_{N,m,b,\ell})
 \le
 \rankpsd^{\R}(M_{N,m,b})
 \sum_{j=0}^{\ell-1}N^j+1.
 \label{eq:amplified-rank-upper}
\end{equation}
\end{proposition}

\begin{proof}
We first compute the pseudo-expectation of the amplified matrix and then bound
its PSD rank from above.

\paragraph{Formal evaluation of every univariate polynomial in $H$.}
Let $H=\sum_iX_i$ and $t=m/2$.  From
\eqref{eq:falling-factorial-boolean} and
\eqref{eq:grigoriev-moments},
\begin{align*}
 \mathcal G_m((H)_j)
 &=j!\sum_{\abs S=j}\mathcal G_m(X^S)\\
 &=j!\binom mj
   \frac{\binom tj}{\binom mj}\\
 &=j!\binom tj
 =(t)_j.
\end{align*}
Because the falling factorials form a basis for univariate polynomials of
degree at most $m$, every such polynomial $p$ satisfies
\begin{equation}
 \mathcal G_m(p(H))=p(t).
 \label{eq:formal-layer-evaluation}
\end{equation}
Thus $\mathcal G_m$ evaluates low-degree symmetric polynomials exactly as
though $H$ were equal to the nonexistent value $t=m/2$.

\paragraph{Chebyshev growth at the formal negative point.}
At the formal layer, the shifted knapsack polynomial has value
\[
 u_*=b-\frac1{4m^2}\le-\frac1{16m^2}.
\]
Write $v=-u_*>0$ and choose $y\ge0$ so that
$y=\operatorname{arsinh}\sqrt v$.  Then
$\sinh^2y=v$ and
\[
 2u_*-1=-1-2v=-\cosh(2y).
\]
Since $\ell$ is odd, $T_\ell$ is odd; using
\eqref{eq:chebyshev-hyperbolic},
\begin{align}
 P_\ell(u_*)
 &=\frac{1+T_\ell(-\cosh(2y))}{2}
 =\frac{1-\cosh(2\ell y)}2
 \notag\\
 &=-\sinh^2(\ell y).
 \label{eq:chebyshev-negative-exact}
\end{align}
Convexity of $\sinh$ and $\sinh(0)=0$ imply
$\sinh(\ell y)\ge\ell\sinh y$ for every integer $\ell\ge1$.  Since
$v\ge1/(16m^2)$, equation
\eqref{eq:chebyshev-negative-exact} gives
\begin{equation}
 P_\ell(u_*)
 \le-\ell^2v
 \le-\frac{\ell^2}{16m^2}
 =-g_{m,\ell}.
 \label{eq:chebyshev-formal-gap}
\end{equation}

The polynomial $f_m(H)+b$ has degree two in $H$, so
$P_\ell(f_m(H)+b)$ has degree at most $2\ell\le m$.  For each fixed row
$S$, the restriction $x_S$ is uniform on $\{0,1\}^m$.  Therefore
\eqref{eq:density-represents-G} and
\eqref{eq:formal-layer-evaluation} imply
\begin{align}
 L_{\mathcal D_m}(P_\ell[M_{N,m,b}])
 &=\mathcal G_m\bigl(P_\ell(f_m(H)+b)\bigr)
 \notag\\
 &=P_\ell\left(b-\frac1{4m^2}\right)
 =P_\ell(u_*).
 \label{eq:amplified-functional-value}
\end{align}

\paragraph{Normalization and the LRS lower bound.}
On the Boolean cube,
\[
 0\le M_{N,m,b}
 \le\frac{m^2-1}{4m^2}+\frac3{16m^2}<1.
\]
Because $P_\ell$ maps $[0,1]$ into $[0,1]$, equation
\eqref{eq:amplified-matrix} gives
$0\le Z_{N,m,b,\ell}\le1$.  Moreover,
$2\ell\le m$ implies $g_{m,\ell}\le1/64$.  Every entry of
$P_\ell[M]$ is nonnegative, so the added constant matrix gives
\begin{equation}
 \norm{Z_{N,m,b,\ell}}_1
 \ge\frac{g_{m,\ell}/2}{1+g_{m,\ell}/2}
 =\frac{g_{m,\ell}}{2+g_{m,\ell}}
 \ge\frac{g_{m,\ell}}3.
 \label{eq:amplified-l1}
\end{equation}
Since $L_{\mathcal D_m}(\mathbf J)=\E\mathcal D_m=1$, equations
\eqref{eq:chebyshev-formal-gap} and
\eqref{eq:amplified-functional-value} yield
\begin{align}
 L_{\mathcal D_m}(Z_{N,m,b,\ell})
 &\le
 \frac{-g_{m,\ell}+g_{m,\ell}/2}
      {1+g_{m,\ell}/2}
 \notag\\
 &=-\frac{g_{m,\ell}}{2+g_{m,\ell}}
 <-\frac{g_{m,\ell}}4.
 \label{eq:amplified-negative-functional}
\end{align}

Apply \eqref{eq:lrs-contrapositive} with pseudo-degree $r=m$ and
\[
 \eta=\frac{g_{m,\ell}}4
 =\frac{\ell^2}{64m^2}.
\]
By \eqref{eq:amplified-l1}, $\norm Z_1\ge\eta$, and by
\cref{lem:explicit-pseudodensity},
$\norm{\mathcal D_m}_\infty\le\sqrt m$.  Hence
\begin{equation}
 \rankpsd^{\R}(Z_{N,m,b,\ell})
 >
 \sqrt\eta
 \left(
  \frac{c_{\mathrm{LRS}}\eta N}
       {m^{7/2}\log N}
 \right)^{m/4}
 \left(\frac\eta{\sqrt m}\right)^{3/2}.
 \label{eq:amplified-lrs-before-substitution}
\end{equation}
Substituting $\eta=\ell^2/(64m^2)$ gives
\[
 \sqrt\eta
 \left(\frac\eta{\sqrt m}\right)^{3/2}
 =\frac{\ell}{8m}
  \frac{\ell^3}{512m^{15/4}}
 =\frac{\ell^4}{4096m^{19/4}},
\]
and the term inside the $m/4$ power becomes
\[
 \frac{c_{\mathrm{LRS}}\ell^2N}
      {64m^{11/2}\log N}.
\]
This proves \eqref{eq:amplified-rank-lower}.

\paragraph{PSD-rank cost of amplification.}
Put
\[
 C=Q_\ell[M_{N,m,b}].
\]
For a fixed row $S$, the function $x\mapsto M_{N,m,b}(S,x)$ is a Boolean
polynomial of degree two.  Since
$\deg Q_\ell=(\ell-1)/2$, the corresponding row of $C$ has degree at most
$\ell-1$.  It can therefore be expanded as
\[
 C(S,x)=\sum_{\substack{U\subseteq[N]\\\abs U\le\ell-1}}
 a_{S,U}x^U.
\]
All rows lie in the span of the coordinate functions $x\mapsto x^U$ with
$\abs U\le\ell-1$.  Consequently,
\begin{equation}
 \operatorname{rank}(C)
 \le\sum_{j=0}^{\ell-1}\binom Nj
 \le\sum_{j=0}^{\ell-1}N^j
 =:q.
 \label{eq:feature-rank}
\end{equation}
By \eqref{eq:squared-rank-factorization},
$C^{\odot2}$ has real PSD rank at most $q$.  The polynomial identity
\eqref{eq:chebyshev-square} holds entrywise, so
\[
 P_\ell[M_{N,m,b}]
 =M_{N,m,b}\odot C^{\odot2}.
\]
The Hadamard-product inequality \eqref{eq:psdrank-hadamard} now gives
\[
 \rankpsd^{\R}(P_\ell[M_{N,m,b}])
 \le q\rankpsd^{\R}(M_{N,m,b}).
\]
Adding the positive constant matrix costs at most one dimension by
\eqref{eq:psdrank-sum}, and the final positive rescaling does not change PSD
rank.  This proves \eqref{eq:amplified-rank-upper}.
\end{proof}

\section{Product states and bounded block-positive witnesses}
\label{sec:witnesses}

The next proposition constructs the product states and block-positive
operators used in the PSD-rank lower bound.  The state $\sigma_S$ encodes a
subset $S$ by placing half of its mass on a
distinguished basis vector and spreading the remaining half uniformly over
$S$.  The witness $W_T$ reads the intersection size $\abs{S\cap T}$ through
diagonal projectors.  A penalty term forces the relevant complex coordinates
toward the Boolean values $0$ and $1$.  Adding a multiple of $I-F$ then extends
nonnegativity from vectors $z\otimes z$ to arbitrary product vectors.
The important quantitative point is that the final operator norm is $O(m)$,
independent of the ambient dimension $N+1$ and of the size of $T$.

\begin{proposition}[Product states and bounded witnesses]
\label[proposition]{prop:witnesses}
There is a universal constant $C_W>0$ with the following property.  Let
$m\ge3$ be odd, let $N\ge m$, and put $H=\C^{N+1}$.  There are pure product
states
\[
 \sigma_S\in\Sep(H:H),
 \qquad S\in\binom{[N]}m,
\]
and block-positive Hermitian operators
\[
 W_T\in\Herm(H\otimes H),
 \qquad T\subseteq[N],
\]
such that
\begin{align}
 \Tr(W_T\sigma_S)
 &=f_m(\one_T|_S)+\frac1{8m^2},
 \label{eq:witness-entry}\\
 \norm{W_T}_\infty&\le C_Wm.
 \label{eq:witness-norm}
\end{align}
\end{proposition}

\begin{proof}
Fix an orthonormal basis $e_0,e_1,\ldots,e_N$ of $H$.  We proceed in four
steps: define the candidate operator, prove nonnegativity on diagonal product
vectors, extend to arbitrary product vectors, and finally evaluate the
witness on the chosen states.

\paragraph{Step 1: states, penalty vectors, and the norm bound.}
For $S\in\binom{[N]}m$, set
\begin{equation}
 \psi_S=\frac{\sqrt m\,e_0+\sum_{i\in S}e_i}{\sqrt{2m}},
 \qquad
 \sigma_S=\proj{\psi_S}\otimes\proj{\psi_S}.
 \label{eq:balanced-product-state}
\end{equation}
The numerator defining $\psi_S$ has squared norm $m+m=2m$, so
$\psi_S$ is a unit vector and $\sigma_S$ is a pure product state.

For $T\subseteq[N]$, define
\[
 E=\proj{e_0},
 \qquad
 P_T=\sum_{i\in T}\proj{e_i},
\]
and let $F$ be the swap operator on $H\otimes H$.  Introduce the vectors
\[
 r_i=e_i\otimes e_i-\frac1{\sqrt m}e_i\otimes e_0
\]
and the positive semidefinite operators
\[
 R_T=\sum_{i\in T}\proj{r_i},
 \qquad
 \widetilde{R}_T=\frac12(R_T+FR_TF).
\]
For distinct $i,j$, the vectors $r_i$ and $r_j$ are orthogonal, and
\[
 \norm{r_i}^2=1+\frac1m.
\]
Thus $R_T$ has operator norm $1+1/m$.  The same is true of $FR_TF$, and
convexity of the operator norm gives
\begin{equation}
 \norm{\widetilde{R}_T}_\infty
 \le1+\frac1m
 \le\frac43.
 \label{eq:Rtilde-norm}
\end{equation}

Put
\[
 \vartheta=\frac12-\frac1{2\sqrt2}>0,
 \qquad
 \lambda_m=\frac{20m}{\vartheta^2},
\]
and define
\begin{align}
 B_T={}&4\left[
 P_T\otimes P_T
 -\frac12(P_T\otimes E+E\otimes P_T)
 +\frac{m^2-1}{4m^2}E\otimes E
 \right]
 \notag\\
 &\quad+\frac1{2m^2}E\otimes E
 +\lambda_m\widetilde{R}_T.
 \label{eq:B-definition}
\end{align}
The supports of
$P_T\otimes P_T$, $P_T\otimes E$, $E\otimes P_T$, and $E\otimes E$ are
pairwise orthogonal.  After accounting for their coefficients, the operator
in the square brackets, including the outer factor $4$, has norm at most
$4$.  Therefore, by \eqref{eq:Rtilde-norm},
\[
 \norm{B_T}_\infty
 \le4+\frac1{2m^2}+\frac43\lambda_m
 \le\frac{30m}{\vartheta^2}.
\]
The last inequality is deliberately loose but uniform for every $m\ge3$.
Set
\begin{equation}
 L_m=\frac{30m}{\vartheta^2}.
 \label{eq:Lm-definition}
\end{equation}
Every term in \eqref{eq:B-definition} is invariant under interchange of the
two tensor factors, so
\begin{equation}
 [B_T,F]=0,
 \qquad
 \norm{B_T}_\infty\le L_m.
 \label{eq:B-norm}
\end{equation}

\paragraph{Step 2: nonnegativity on vectors $z\otimes z$.}
Write
\[
 z=\sum_{i=0}^Nz_ie_i,
 \qquad
 \alpha=\abs{z_0}^2,
 \qquad
 p=\sum_{i\in T}\abs{z_i}^2.
\]
The diagonal tensor terms have expectations
\begin{align*}
 \inner{z\otimes z}{(P_T\otimes P_T)(z\otimes z)}&=p^2,\\
 \inner{z\otimes z}{(P_T\otimes E)(z\otimes z)}&=p\alpha,\\
 \inner{z\otimes z}{(E\otimes P_T)(z\otimes z)}&=\alpha p,\\
 \inner{z\otimes z}{(E\otimes E)(z\otimes z)}&=\alpha^2.
\end{align*}
Moreover, $F(z\otimes z)=z\otimes z$, so $R_T$ and $FR_TF$ have the same
expectation on this vector.  Since
\[
 \inner{r_i}{z\otimes z}
 =z_i\left(z_i-\frac{z_0}{\sqrt m}\right)
\]
up to complex conjugation, which disappears after taking absolute values, we
obtain
\begin{align}
 \inner{z\otimes z}{B_T(z\otimes z)}
 ={}&4\left[
       p^2-p\alpha+\frac{m^2-1}{4m^2}\alpha^2
      \right]
   +\frac{\alpha^2}{2m^2}
 \notag\\
 &\quad+\lambda_m\sum_{i\in T}
 \abs{z_i\left(z_i-\frac{z_0}{\sqrt m}\right)}^2
 \notag\\
 ={}&4\left[
       \left(p-\frac\alpha2\right)^2
       -\frac{\alpha^2}{4m^2}
      \right]
   +\frac{\alpha^2}{2m^2}
 \notag\\
 &\quad+\lambda_m\sum_{i\in T}
 \abs{z_i\left(z_i-\frac{z_0}{\sqrt m}\right)}^2.
 \label{eq:diagonal-B}
\end{align}
If $\alpha=0$, then
\[
 \inner{z\otimes z}{B_T(z\otimes z)}
 =4p^2+\lambda_m\sum_{i\in T}\abs{z_i}^4\ge0.
\]
Assume henceforth that $\alpha>0$, and rescale the relevant coordinates by
setting
\[
 w_i=\frac{\sqrt m\,z_i}{z_0},
 \qquad
 s=\sum_{i\in T}\abs{w_i}^2,
 \qquad
 Q=\sum_{i\in T}\abs{w_i(w_i-1)}^2.
\]
Then $p=\alpha s/m$, and the penalty term equals
$\alpha^2Q/m^2$.  Substitution into \eqref{eq:diagonal-B} gives
\begin{equation}
 \inner{z\otimes z}{B_T(z\otimes z)}
 =\frac{\alpha^2}{m^2}
 \left\{
  4\left[
    \left(s-\frac m2\right)^2-\frac18
   \right]
  +\lambda_mQ
 \right\}.
 \label{eq:reduced-diagonal-B}
\end{equation}
If the first term in braces is nonnegative, there is nothing to prove.  We
may therefore assume
\begin{equation}
 \abs{s-\frac m2}<\frac1{2\sqrt2}.
 \label{eq:s-window}
\end{equation}
We show that in this narrow window the penalty $Q$ is necessarily large.

For every $i\in T$, choose $b_i\in\{0,1\}$ nearest to $w_i$ in the complex
plane; ties may be broken arbitrarily.  The two distances
$\abs{w_i}$ and $\abs{w_i-1}$ cannot both be smaller than $1/2$, because
their sum is at least the distance between $0$ and $1$.  Therefore
\[
 \max\{\abs{w_i},\abs{w_i-1}\}\ge\frac12.
\]
Since the product of the two distances is $\abs{w_i(w_i-1)}$, the smaller
distance satisfies
\begin{equation}
 \abs{w_i-b_i}^2
 \le4\abs{w_i(w_i-1)}^2.
\end{equation}
Summing over $i$ gives
\begin{equation}
 \sum_{i\in T}\abs{w_i-b_i}^2\le4Q.
 \label{eq:rounding-Q}
\end{equation}
Let $k=\sum_{i\in T}b_i$.  If $b_i=1$, then
$\abs{w_i-1}\le\abs{w_i}$, which implies
$\operatorname{Re}w_i\ge1/2$ and hence $\abs{w_i}\ge1/2$.  It follows that
\begin{equation}
 k\le4\sum_{i\in T}\abs{w_i}^2=4s.
 \label{eq:k-versus-s}
\end{equation}

We now compare the real number $s$ with the integer $k$.  Since
$b_i^2=b_i$,
\begin{align*}
 \abs{s-k}
 &=\abs*{\sum_{i\in T}(\abs{w_i}^2-b_i^2)}\\
 &\le\sum_{i\in T}
   \abs{w_i-b_i}(\abs{w_i}+b_i)\\
 &\le
 \left(\sum_{i\in T}\abs{w_i-b_i}^2\right)^{1/2}
 \left(\sum_{i\in T}(\abs{w_i}+b_i)^2\right)^{1/2}.
\end{align*}
The first factor is at most $\sqrt{4Q}$ by
\eqref{eq:rounding-Q}.  For the second,
\[
 \sum_{i\in T}(\abs{w_i}+b_i)^2
 \le2s+2k
 \le10s
 <10m,
\]
where we used \eqref{eq:k-versus-s} and the fact that
\eqref{eq:s-window} implies $s<m$ for $m\ge3$.  Hence
\begin{equation}
 \abs{s-k}\le\sqrt{40mQ}.
 \label{eq:s-k-Q}
\end{equation}

Because $m/2$ is a half-integer and $k$ is an integer, the window
\eqref{eq:s-window} gives
\[
 \abs{s-k}
 \ge\abs{\frac m2-k}-\abs{s-\frac m2}
 >\frac12-\frac1{2\sqrt2}
 =\vartheta.
\]
Combining this with \eqref{eq:s-k-Q} gives
\begin{equation}
 Q\ge\frac{\vartheta^2}{40m},
 \qquad
 \lambda_mQ\ge\frac12.
 \label{eq:Q-lower}
\end{equation}
The first term in braces in \eqref{eq:reduced-diagonal-B} is always at least
$-1/2$, because the square is nonnegative.  \Cref{eq:Q-lower} shows that the
penalty compensates for this possible
negative contribution.  We have proved
\begin{equation}
 \inner{z\otimes z}{B_T(z\otimes z)}\ge0
 \qquad\text{for every }z\in H.
 \label{eq:diagonal-nonnegative}
\end{equation}

\paragraph{Step 3: extension from $z\otimes z$ to $u\otimes v$.}
Suppose a Hermitian
operator $B$ satisfies
\[
 [B,F]=0,
 \qquad
 \norm B_\infty\le L,
 \qquad
 \inner{z\otimes z}{B(z\otimes z)}\ge0
 \quad\text{for all }z.
\]
Then
\begin{equation}
 B+L(I-F)
 \label{eq:swap-extension}
\end{equation}
is block-positive and has norm at most $3L$.

It suffices to test unit vectors $u,v$.  Change the phase of $v$ so that
$q=\inner uv\in[0,1]$ is real, and put
\[
 \tau=1-q^2.
\]
If $q=1$, then $u=v$ after the phase choice and diagonal nonnegativity
applies.  Suppose $q<1$, and define
\[
 y=\frac{u+v}{\sqrt{2(1+q)}},
 \qquad
 w=\frac{u-v}{\sqrt{2(1-q)}},
\]
which are orthonormal.  With
\[
 \alpha=\frac{1+q}{2},
 \qquad
 \beta=\frac{1-q}{2},
\]
we have
\[
 u=\sqrt\alpha\,y+\sqrt\beta\,w,
 \qquad
 v=\sqrt\alpha\,y-\sqrt\beta\,w.
\]
For $x=u\otimes v$, its symmetric and antisymmetric parts are
\begin{align*}
 x_+=\frac{x+Fx}{2}
 &=\alpha y\otimes y-\beta w\otimes w,\\
 x_-=\frac{x-Fx}{2}
 &=\sqrt{\alpha\beta}
   (w\otimes y-y\otimes w).
\end{align*}
Since $2\alpha\beta=\tau/2$,
\begin{equation}
 \norm{x_-}^2=2\alpha\beta=\frac\tau2.
 \label{eq:antisymmetric-norm}
\end{equation}
The commutation relation $[B,F]=0$ makes the symmetric and antisymmetric
subspaces invariant, so the cross term between $x_+$ and $x_-$ vanishes.
For the symmetric part, diagonal nonnegativity removes the two diagonal
terms, while the mixed matrix element is bounded by $L$:
\begin{align*}
 \inner{x_+}{Bx_+}
 &\ge-2\alpha\beta
 \abs{\inner{y\otimes y}{B(w\otimes w)}}\\
 &\ge-2\alpha\beta L
 =-\frac L2\tau.
\end{align*}
For the antisymmetric part, \eqref{eq:antisymmetric-norm} and the operator
norm bound give
\[
 \inner{x_-}{Bx_-}
 \ge-L\norm{x_-}^2
 =-\frac L2\tau.
\]
Thus
\begin{equation}
 \inner{u\otimes v}{B(u\otimes v)}\ge-L\tau.
 \label{eq:B-product-lower}
\end{equation}
On the other hand,
\begin{align}
 \inner{u\otimes v}{(I-F)(u\otimes v)}
 &=1-\abs{\inner uv}^2
 =\tau.
 \label{eq:swap-penalty-product}
\end{align}
\Cref{eq:B-product-lower,eq:swap-penalty-product} prove block positivity of
\eqref{eq:swap-extension}.  Its norm is at most
$L+L\norm{I-F}_\infty=3L$ by \eqref{eq:swap-norm}.

Apply this argument to $B_T$ using \eqref{eq:B-norm} and
\eqref{eq:diagonal-nonnegative}, and define
\begin{equation}
 W_T=B_T+L_m(I-F).
 \label{eq:W-definition}
\end{equation}
Then $W_T$ is block-positive and
\[
 \norm{W_T}_\infty
 \le3L_m
 =\frac{90m}{\vartheta^2}.
\]
Thus \eqref{eq:witness-norm} holds with the universal constant
$C_W=90/\vartheta^2$.

\paragraph{Step 4: evaluation on the encoded product states.}
Let
\[
 r=\abs{S\cap T}.
\]
From \eqref{eq:balanced-product-state},
\begin{equation}
 \inner{\psi_S}{E\psi_S}=\frac12,
 \qquad
 \inner{\psi_S}{P_T\psi_S}=\frac r{2m}.
 \label{eq:EP-expectations}
\end{equation}
The penalty term vanishes on $\psi_S\otimes\psi_S$.  Indeed, for each
$i\in T$,
\[
 \inner{r_i}{\psi_S\otimes\psi_S}
 =\one_{\{i\in S\}}
 \left(\frac1{2m}-\frac1{\sqrt m}\frac1{2\sqrt m}\right)
 =0.
\]
Hence $R_T(\psi_S\otimes\psi_S)=0$, and the same is true for
$FR_TF$.  Also
$(I-F)(\psi_S\otimes\psi_S)=0$ because this vector is symmetric.
Substituting \eqref{eq:EP-expectations} into
\eqref{eq:B-definition} therefore gives
\begin{align*}
 \Tr(W_T\sigma_S)
 &=4\left[
   \left(\frac r{2m}\right)^2
   -\frac12\left(
      \frac r{2m}\frac12+\frac12\frac r{2m}
    \right)
   +\frac{m^2-1}{4m^2}\frac14
  \right]
  +\frac1{2m^2}\frac14\\
 &=\frac{r^2}{m^2}-\frac rm
   +\frac{m^2-1}{4m^2}
   +\frac1{8m^2}\\
 &=\frac1{m^2}
   \left[
    \left(r-\frac m2\right)^2-\frac14
   \right]
   +\frac1{8m^2}\\
 &=f_m(\one_T|_S)+\frac1{8m^2}.
\end{align*}
This proves \eqref{eq:witness-entry} and completes the proof.
\end{proof}

\section{Semidefinite extension complexity of separability}
\label{sec:hsep-extension-lower-bounds}

The slack-matrix/LRS approach of this section follows HNW's application
to $h_{\Sep}$~\cite[Sec.~5.2 and proof of Thm.~5.6]{harrowLimitationsSemidefinitePrograms2019}.
HNW obtain their hard submatrix by applying LRS to the acceptance-probability
objective of a logarithmic-proof $\QMA(2)$ protocol and then embedding the
resulting $\QMA(2)$-Honest optimization problem into $h_{\Sep}$.  Our new
replacement for this protocol-based realization is the direct bounded
knapsack-witness construction of \cref{prop:witnesses}: the operators $W_T$
and product states $\sigma_S$ directly realize the bounded-knapsack values,
and the resulting slack entries form the shifted knapsack pattern matrix.

We now prove a lower bound for semidefinite formulations that compute the
maximum measurement value over separable states:
\[
 h_{\Sep(d:d)}(Q)
 =\max_{\sigma\in\Sep(\C^d:\C^d)}\Tr(Q\sigma).
\]
The same SDP feasible region and product-state representation must be used
for every observable $Q$.  From any such formulation we extract a
nonnegative matrix with entries
$c-\Phi_{Q_T}(\iota(\sigma_S))$, indexed by the witnesses and product states
of \cref{prop:witnesses}.  Each entry is the difference, or
\emph{slack}, between an upper threshold $c$ and the value given by a
selected product state.  This matrix equals the shifted
knapsack pattern matrix $M_{N,m,b}$ up to a positive scalar, and its PSD
rank is at most one more than the size of the formulation.  The lower bounds of
\Cref{sec:knapsack} then apply.

We specialize HNW's embedded-reduction framework to $h_{\Sep}$ below,
with the range of each objective constrained on the feasible set.\footnote{HNW's
Definition~5.1 writes an affine objective with codomain $[0,1]$ on the
entire ambient matrix space.  Taken literally, this would force the
objective to be constant.  Here the objective is affine and real-valued
on the ambient space, and its values are required to lie in $[0,1]$ only
on the feasible set.}

\begin{definition}[SDP extended formulation in the HNW framework]
\label[definition]{def:hnw-hsep-ef}
Fix $0\le s<c\le1$.  A size-$r$, $(c,s)$-approximate SDP extended
formulation for $h_{\Sep(d:d)}$ consists of the following data:
\begin{enumerate}[label=\textup{(\roman*)},leftmargin=*]
 \item \emph{One fixed set of SDP constraints:} a nonempty set
 \[
  \mathcal P=\mathcal A\cap\mathbb S_+^r,
  \qquad
  \mathbb S_+^r
  =\{Y\in\R^{r\times r}:Y=Y^{\mathsf T},\ Y\succeq0\},
 \]
 where $\mathcal A$ is an affine subspace.  The set $\mathcal P$ is chosen
 once and does not depend on the measurement operator $Q$.
 \item \emph{Affine objectives:} for every two-outcome measurement operator
 $Q$, meaning $0\preceq Q\preceq I$, an affine functional
 $\Phi_Q:\mathbb S^r\to\R$ satisfying
 $0\le\Phi_Q(Y)\le1$ for every $Y\in\mathcal P$.
 \item \emph{An objective-independent embedding:} a map $\iota$ from the
 pure product states on $\C^d\otimes\C^d$ into $\mathcal P$ that does not
 depend on $Q$.
\end{enumerate}
For every such $Q$, these data satisfy
\begin{align}
 \Phi_Q(\iota(\tau))&=\Tr(Q\tau)
 &&\text{for every pure product state $\tau$},
 \label{eq:hnw-embedding}\\
 h_{\Sep(d:d)}(Q)\le s
 &\Longrightarrow
 \sup_{Y\in\mathcal P}\Phi_Q(Y)\le c.
 \label{eq:hnw-soundness}
\end{align}
\end{definition}

To relate this definition to HNW's notation, take their optimization
problem $A$ to be $h_{\Sep}$ and their SDP $B$ to have feasible set
$P^B=\mathcal P$.  Their reduction sends the instance
$\tau\mapsto\Tr(Q\tau)$ to $\Phi_Q$, and their embedding map $E$ is
$\iota$.  Equation~\eqref{eq:hnw-embedding} is the embeddedness condition
of their Definition~3.11, and \eqref{eq:hnw-soundness} is its
$(s^B,s^A)=(c,s)$ approximation condition, as used in their
Definition~5.2~\cite{harrowLimitationsSemidefinitePrograms2019}.
The data $\mathcal P$ and $\iota$ may depend on $d,c,s$, but not on $Q$.
We require each $\Phi_Q$ to be affine in the SDP variable $Y$; we do not
require the assignment $Q\mapsto\Phi_Q$ to be affine.
The embedding makes completeness automatic: maximizing over the embedded
product states already yields $h_{\Sep}(Q)$.

\begin{lemma}[Affine slack factorization]
\label[lemma]{lem:affine-slack-factorization}
Let
\[
 \mathcal P=\mathcal A\cap\mathbb S_+^r
\]
be a nonempty spectrahedron.  Let $Y_j\in\mathcal P$, for $j\in J$, and let
$\ell_i:\mathcal P\to\mathbb R_+$, for $i\in I$, be affine functions that are
nonnegative on all of $\mathcal P$.  Then the matrix
\[
 S(i,j)=\ell_i(Y_j)
\]
satisfies
\[
 \rankpsd^{\mathbb R}(S)\le r+1.
\]
\end{lemma}

\begin{proof}
This is the standard lift-to-factorization implication for the PSD cone;
see~\cite[Thm.~2.4 and Cor.~2.6]{gouveiaLiftsConvexSets2013} and, in the
SDP-relaxation setting,~\cite[Prop.~6.1]{leeLowerBoundsSize2015}.  We include
the short duality argument to make the extra affine dimension explicit.

Let
\[
 K=\operatorname{span}\{\operatorname{range}(Y):Y\in\mathcal P\}
\]
and put $k=\dim K\le r$.  After compressing all matrices to $K$, we may regard
$\mathcal P$ as a spectrahedron in $\mathbb S_+^k$ containing a positive
definite point.  Indeed, choose finitely many feasible matrices whose ranges
span $K$ and average them; for PSD matrices, the kernel of a sum is the
intersection of the kernels.  If $k=0$, then $\mathcal P=\{0\}$, and the
$1\times1$ factors $[\ell_i(0)]$ and $[1]$ prove the claim.  Hence assume
$k\ge1$.

Write the affine constraint set of the compressed spectrahedron as
\[
 \mathcal A'
 =
 \left\{
  Y\in\mathbb S^k:
  \Tr(C_aY)=b_a,\quad a=1,\ldots,m
 \right\},
\]
after removing redundant equations.  The positive definite feasible point
has a relative neighborhood in $\mathcal A'$ that remains positive definite,
so $\mathcal A'=\operatorname{aff}(\mathcal P)$.  Extend each $\ell_i$ to this
affine hull and write
\[
 \ell_i(Y)=\Tr(C_iY)+\beta_i
 \qquad (Y\in\mathcal A').
\]
Consider the SDP
\[
 p_i
 =
 \inf\left\{
  \Tr(C_iY)+\beta_i:
  Y\succeq0,
  \Tr(C_aY)=b_a\ \text{for all }a
 \right\}.
\]
Since $\mathcal P$ is nonempty and $\ell_i$ is nonnegative on it, we have
$0\le p_i<\infty$.  The compressed spectrahedron contains a positive
definite feasible point, so the primal SDP satisfies Slater's condition.
Slater's theorem gives strong duality and attainment of the dual optimum;
see~\cite[Sec.~5.9.1, pp.~265--266]{boydConvexOptimization2004}.  Hence there
are numbers $z_{i,a}$ such that
\[
 A_i:=C_i-\sum_{a=1}^m z_{i,a}C_a\succeq0
\]
and
\[
 \alpha_i
 :=
 \beta_i+\sum_{a=1}^m z_{i,a}b_a
 =
 p_i
 \ge0.
\]
Consequently, for every $Y\in\mathcal P$,
\[
 \ell_i(Y)
 =
 \Tr(A_iY)+\alpha_i.
\]

Define PSD matrices of size $k+1$ by
\[
 \widehat A_i=A_i\oplus[\alpha_i],
 \qquad
 \widehat B_j=Y_j\oplus[1].
\]
Then
\[
 S(i,j)
 =\ell_i(Y_j)
 =\Tr(\widehat A_i\widehat B_j).
\]
Thus $S$ has a real PSD factorization of size $k+1\le r+1$.
\end{proof}

\begin{proposition}[The knapsack matrix is an $h_{\Sep}$ slack submatrix]
\label[proposition]{prop:hsep-slack-submatrix}
Let $m\ge3$ be odd, let $N=d-1\ge m$, and let the product states
$\sigma_S$ and witnesses $W_T$ be those of
\cref{prop:witnesses}.  Let $0<s<c\le1$, and put
\begin{equation}
 u=\min\{s,1-s\},
 \qquad
 \varrho=\frac{c-s}{u},
 \qquad
 L=C_Wm.
 \label{eq:hsep-relative-gap}
\end{equation}
Define, for every $T\subseteq[N]$, the two-outcome measurement operator
\begin{equation}
 Q_T=sI_{d^2}-\frac{u}{L}W_T.
 \label{eq:hsep-hard-effect}
\end{equation}
Then
\begin{equation}
 0\preceq Q_T\preceq I_{d^2},
 \qquad
 h_{\Sep(d:d)}(Q_T)\le s.
 \label{eq:hsep-hard-effect-soundness}
\end{equation}
Set
\begin{equation}
 b=\frac1{8m^2}+\varrho L.
 \label{eq:hsep-pattern-shift}
\end{equation}
If a $(c,s)$-approximate SDP extended formulation for
$h_{\Sep(d:d)}$ has size $r$, then
\begin{equation}
 \rankpsd^{\R}(M_{N,m,b})\le r+1.
 \label{eq:hsep-pattern-rank-upper}
\end{equation}
\end{proposition}

\begin{proof}
The norm bound $\norm{W_T}_\infty\le L$ gives
\[
 (s-u)I\preceq Q_T\preceq(s+u)I.
\]
The definition $u=\min\{s,1-s\}$ implies
$0\le s-u\le s+u\le1$, so $Q_T$ is a valid measurement operator.
If $\tau$ is separable,
block positivity gives $\Tr(W_T\tau)\ge0$, and therefore
\[
 \Tr(Q_T\tau)
 =s-\frac uL\Tr(W_T\tau)
 \le s.
\]
This proves \eqref{eq:hsep-hard-effect-soundness}.

Let $\mathcal P$, $\Phi_Q$, and $\iota$ describe the extended formulation.
For $S\in\binom{[N]}m$, put $Y_S=\iota(\sigma_S)$.  For each
$T\subseteq[N]$, define
\[
 \ell_T(Y)=c-\Phi_{Q_T}(Y).
\]
By \eqref{eq:hsep-hard-effect-soundness} and the soundness condition
\eqref{eq:hnw-soundness}, $\ell_T$ is nonnegative on all of $\mathcal P$.
The embedding identity \eqref{eq:hnw-embedding}, the relation
$c-s=u\varrho$, and \cref{prop:witnesses} give
\begin{align}
 \ell_T(Y_S)
 &=c-\Tr(Q_T\sigma_S)
 \notag\\
 &=c-s+\frac uL\Tr(W_T\sigma_S)
 \notag\\
 &=\frac uL
 \left(
  f_m(\one_T|_S)+\frac1{8m^2}+\varrho L
 \right)
 \notag\\
 &=\frac uL M_{N,m,b}(S,\one_T).
 \label{eq:hsep-slack-identity}
\end{align}

Apply \cref{lem:affine-slack-factorization} to the feasible points
$Y_S=\iota(\sigma_S)$ and the nonnegative affine functions
\[
 \ell_T(Y)=c-\Phi_{Q_T}(Y).
\]
Together with \eqref{eq:hsep-slack-identity}, this gives
\[
 \rankpsd^{\mathbb R}(M_{N,m,b})\le r+1,
\]
since transposition and multiplication by a positive scalar do not change
PSD rank.  This is \eqref{eq:hsep-pattern-rank-upper}.
\end{proof}

The preceding proposition gives the improved parameters.  The Boolean
assignment $\one_T$ is encoded by the objective $Q_T$, while the subset $S$
is encoded by the product state $\sigma_S$.  Both use local dimension
$N+1=d$, so we do not need to embed an $m$-variable protocol into an
exponentially larger instance.

\paragraph{Parameters.}
For $d\ge3$, put $N=d-1$ and $x=\log N$.  For $\eta\ge0$, define
\begin{align}
 M_\eta
 &:=\min\left\{
  \eta^{-1/3},
  \left(\frac{N}{x^3}\right)^{2/7}
 \right\},
 \label{eq:hsep-endpoint-scale}\\
 R_\eta
 &:=\frac{N}{M_\eta^{7/2}x^3}.
 \label{eq:hsep-endpoint-ratio}
\end{align}
When $\eta=0$, we use the convention $\eta^{-1/3}=+\infty$.
By construction, $R_\eta\ge1$.

\begin{theorem}[Lower bound for threshold SDP formulations of $h_{\Sep}$]
\label{thm:hsep-endpoint}
\normalfont
There are universal constants $c_{\mathrm{end}}>0$, $0<\varrho_0\le1$, and
$d_0\in\mathbb N$ such that the following holds.  Let $d\ge d_0$, let
$0<s<c\le1$, and suppose that $h_{\Sep(d:d)}$ has a size-$r$
$(c,s)$-approximate HNW extended formulation.  Set
\[
 \varrho=\frac{c-s}{\min\{s,1-s\}}.
\]
If $\varrho\le\varrho_0$, then
\begin{equation}
 \log r\ge
 c_{\mathrm{end}}M_\varrho\log(2+R_\varrho).
 \label{eq:hsep-endpoint-bound}
\end{equation}
\end{theorem}

\noindent\emph{Lean:}
\leanref{Disentangler/ManuscriptAudit.lean}{33}{hsep\_endpoint}.

\begin{proof}
\Cref{prop:hsep-slack-submatrix} gives
\[
 \rankpsd^{\R}(M_{N,m,b})\le r+1,
 \qquad
 b=\frac1{8m^2}+\varrho C_Wm.
\]
Hence there is a universal $c_0>0$ such that
$\varrho m^3\le c_0$ implies $b\le3/(16m^2)$.  Combining this with
\cref{prop:amplified-rank}, we obtain universal constants
$c_1,c_2>0$ such that, for every admissible odd $m$ and every odd
$\ell$ with $2\ell\le m$,
\begin{equation}
 (r+1)S_\ell+1
 \ge
 c_1\ell^4m^{-19/4}
 \left(
  \frac{c_2\ell^2N}{m^{11/2}x}
 \right)^{m/4},
 \qquad
 S_\ell=\sum_{j=0}^{\ell-1}N^j.
 \label{eq:hsep-master-amplified}
\end{equation}

Write $M=M_\varrho$ and $R=R_\varrho$.  Fix a sufficiently large
universal constant $K$, choose a sufficiently small universal
$\gamma>0$, and let $m$ be the largest odd integer at most $\gamma M$.
After decreasing $\varrho_0$ and increasing $d_0$, we may assume
$\gamma M\ge4$ and
\begin{equation}
 \frac{\gamma M}{2}\le m\le\gamma M,
 \qquad m\ge3.
 \label{eq:hsep-endpoint-m-choice}
\end{equation}
Indeed, $M\le\varrho^{-1/3}$, so choosing $\gamma^3\le c_0$ gives
$\varrho m^3\le c_0$.  Moreover,
$M\le(N/x^3)^{2/7}=o(N)$, and hence $2m<N$ once $d_0$ is large enough.
Thus $m$ is an odd integer with $m\ge3$, $2m<N$, and
$\varrho m^3\le c_0$, which are precisely the admissibility conditions
needed for \eqref{eq:hsep-master-amplified}.  Choose
\[
 \ell=
 \begin{cases}
  1,&m\le4Kx,\\
  \text{the largest odd integer at most }m/(Kx),&m>4Kx.
 \end{cases}
\]
Then $2\ell\le m$: this follows from $m\ge3$ in the first case and from
$Kx\ge2$ in the second, after increasing $d_0$ if necessary.  In the
second case, rounding to the largest odd
integer loses at most two, and $m/(Kx)>4$, so
\[
 \frac{m}{2Kx}\le\ell\le\frac{m}{Kx}.
\]
Put $S_\ell=\sum_{j=0}^{\ell-1}N^j$.  When $\ell=1$, $S_\ell=1$; in
the second case, $S_\ell\le N^\ell$.  Consequently,
the logarithmic cost of $S_\ell$ is at most $m/K$.

We also need a lower bound on the base in
\eqref{eq:hsep-master-amplified}.  If $\ell=1$, then $m\le4Kx$ implies
\[
 \frac{\ell^2}{m^{11/2}x}
 \ge\frac{1}{16K^2m^{7/2}x^3}.
\]
If $\ell>1$, the preceding lower bound on $\ell$ gives the same estimate
with $16$ replaced by $4$.  Thus, in both cases, for a universal $c_*>0$,
\begin{equation}
 \log S_\ell\le\frac mK,
 \qquad
 \frac{c_2\ell^2N}{m^{11/2}x}
 \ge
 \frac{c_*N}{K^2m^{7/2}x^3}.
 \label{eq:hsep-optimized-cost-and-base}
\end{equation}
Since $m\le\gamma M$, the definition of $R$ now gives
\begin{equation}
 \log\left(\frac{c_2\ell^2N}{m^{11/2}x}\right)
 \ge
 \log R+\frac72\log\frac1\gamma+\log\frac{c_*}{K^2}.
 \label{eq:hsep-base-log-detail}
\end{equation}
Choose $\gamma$ so that the constant part on the right is sufficiently
positive.  The remaining prefactor costs only logarithmically in $m$:

\begin{equation}
 \log\bigl(c_1\ell^4m^{-19/4}\bigr)\ge-C\log m.
 \label{eq:hsep-prefactor-detail}
\end{equation}

Finally, since $S_\ell\ge1$,
\[
 (r+1)S_\ell+1\le(r+2)S_\ell.
\]
Taking logarithms in \eqref{eq:hsep-master-amplified} therefore gives
\[
 \log(r+2)
 \ge \frac m4
 \log\left(\frac{c_2\ell^2N}{m^{11/2}x}\right)
 -\log S_\ell-C\log m.
\]
Apply \eqref{eq:hsep-optimized-cost-and-base},
\eqref{eq:hsep-base-log-detail}, and \eqref{eq:hsep-prefactor-detail}.
If
\[
 A_\gamma:=\frac72\log\frac1\gamma+\log\frac{c_*}{K^2},
\]
then these estimates give
\[
 \log(r+2)
 \ge\frac m4(\log R+A_\gamma)-\frac mK-C\log m.
\]
Taking $K$ sufficiently large, then $\gamma$ sufficiently small so that
$A_\gamma$ is a large positive constant, and finally increasing the fixed
lower bound on $m$ by decreasing $\varrho_0$ and increasing $d_0$, yields
\[
 \log(r+2)\ge c'm(1+\log R)
\]
for a universal $c'>0$.  Indeed, for a universal $c_m>0$,
\[
 m\ge c_mM,
 \qquad
 \log(2+R)\le\log(3R)\le(\log3)(1+\log R).
\]
Hence there is a universal $c''>0$ such that
\[
 \log(r+2)\ge H,
 \qquad
 H:=c''M\log(2+R).
\]
After decreasing $\varrho_0$ and increasing $d_0$ once more, we may assume
$H\ge2\log3$.  Since $r\ge1$, we have $r+2\le3r$, and therefore
\[
 \log r\ge\log(r+2)-\log3\ge\frac H2.
\]
Reducing the universal constant proves \eqref{eq:hsep-endpoint-bound}.
\end{proof}

\begin{corollary}[Lower bound for additive SDP approximations of $h_{\Sep}$]
\label[corollary]{cor:hsep-uniform-additive}
\normalfont
Suppose a size-$r$ SDP has the common feasible region,
objective-independent product-state embedding, and bounded objectives
required in \cref{def:hnw-hsep-ef}, and suppose that its optimum
satisfies
\begin{equation}
 h_{\Sep(d:d)}(Q)
 \le \operatorname{SDP}(Q)
 \le h_{\Sep(d:d)}(Q)+a
 \qquad\text{for every $Q$ with $0\preceq Q\preceq I$.}
 \label{eq:hsep-uniform-additive}
\end{equation}
If $d\ge d_0$ and $2a\le\varrho_0$, then
\begin{equation}
 \log r\ge
 c_{\mathrm{end}}M_{2a}\log(2+R_{2a}).
 \label{eq:hsep-uniform-additive-bound}
\end{equation}
\end{corollary}

\begin{proof}
For $a>0$, take $s=1/2$ and $c=1/2+a$.  The assumption
$2a\le\varrho_0\le1$ ensures that $c\le1$.  If
$h_{\Sep(d:d)}(Q)\le1/2$, then \eqref{eq:hsep-uniform-additive} gives
$\operatorname{SDP}(Q)\le1/2+a$.  The SDP is therefore a
$(1/2+a,1/2)$-approximate HNW formulation with relative gap $2a$, so
the result follows from \cref{thm:hsep-endpoint}.  When $a=0$,
apply the positive-error statement with an arbitrarily small error
for which the second term in $M_{2a}$ is the minimum.
\end{proof}

\begin{corollary}[Lower bound for spectrahedral approximations of separable states]
\label[corollary]{cor:sep-spectrahedral-shadow}
\normalfont
Let $\mathcal K\subseteq\Dens(\C^d\otimes\C^d)$ be a spectrahedral
shadow of size $r$ such that
\begin{equation}
 \Sep(d:d)\subseteq\mathcal K,
 \qquad
 \sup_{\rho\in\mathcal K}
 \disttr\bigl(\rho,\Sep(d:d)\bigr)\le a.
 \label{eq:sep-shadow-approximation}
\end{equation}
If $d\ge d_0$ and $2a\le\varrho_0$, then
\begin{equation}
 \log r\ge
 c_{\mathrm{end}}M_{2a}\log(2+R_{2a}).
 \label{eq:sep-shadow-lower-bound}
\end{equation}
\end{corollary}

\begin{proof}
Write $\mathcal K=\pi(\mathcal P)$ and choose once and for all a lift
in $\mathcal P$ of every pure product state.  For a measurement operator
$Q$ with $0\preceq Q\preceq I$, use
the affine objective $Y\mapsto\Tr(Q\pi(Y))$.  Since $\mathcal K$
consists of states, this objective takes values in $[0,1]$, and the
chosen lifts give the required objective-independent embedding.  The
set inclusion and trace-distance duality give
\[
 h_{\Sep(d:d)}(Q)\le h_{\mathcal K}(Q)
 \le h_{\Sep(d:d)}(Q)+a.
\]
Thus the lift of $\mathcal K$ satisfies
\eqref{eq:hsep-uniform-additive}, and the result follows from
\cref{cor:hsep-uniform-additive}.
\end{proof}

As a direct consequence of \cref{thm:hsep-endpoint}, for every
fixed $0<\theta<2/7$ there are constants $c_\theta,\eta_\theta>0$ and
$d_\theta\in\mathbb N$ such that
\begin{equation}
 r\ge d^{\,c_\theta\min\{\eta^{-1/3},d^\theta\}}.
 \label{eq:hsep-extension-complexity}
\end{equation}
This holds in all three settings, with $\eta=\varrho$ in
\cref{thm:hsep-endpoint} and $\eta=2a$ in the two corollaries,
whenever $d\ge d_\theta$ and $\eta\le\eta_\theta$.

One can derive this as follows:  Write
\[
 A=\eta^{-1/3},
 \qquad
 B=\left(\frac{N}{x^3}\right)^{2/7},
 \qquad
 q=\min\{A,N^\theta\},
\]
so that $M_\eta=\min\{A,B\}$.  It is enough to prove
$M_\eta\log(2+R_\eta)\ge c_\theta qx$.  Since
\[
 \frac{B}{N^\theta x}
 =\frac{N^{2/7-\theta}}{x^{13/7}}
 \longrightarrow\infty,
\]
we have $B\ge qx$ for all sufficiently large $N$.  If
$M_\eta\ge qx$, then $R_\eta\ge1$ gives the claim.  Otherwise
$M_\eta<qx\le B$, so $M_\eta=A\ge q$, and
$M_\eta<qx\le N^\theta x$.  Consequently,
\[
 R_\eta
 =\frac{N}{M_\eta^{7/2}x^3}
 \ge\frac{N^{1-7\theta/2}}{x^{13/2}}.
\]
Because $\theta<2/7$, this implies
$\log(2+R_\eta)\ge c_\theta x$ for sufficiently large $N$, proving the
claim also in the second case.  Since $N=d-1$ and
$\log N=\Theta(\log d)$, \cref{thm:hsep-endpoint} now gives
\eqref{eq:hsep-extension-complexity} after adjusting constants.

\begin{corollary}[Lower bounds for two-particle bosonic separability]
\label[corollary]{cor:bosonic-hsep}
\normalfont
Define
\[
 \Sep_{\mathrm{bos}}^{(2)}(d)
 =\operatorname{conv}\left\{\proj\psi^{\otimes2}:\psi\in\C^d,
 \ \norm\psi=1\right\}.
\]
This is precisely the set of bipartite separable states supported on the
symmetric subspace $\operatorname{Sym}^2(\C^d)$
~\cite[Thm.~1]{weisDecompositionSymmetricSeparable2020}.
For $0\preceq Q\preceq I$, define
\[
 h_{\mathrm{bos}}^{(2)}(Q)
 =\max_{\sigma\in\Sep_{\mathrm{bos}}^{(2)}(d)}\Tr(Q\sigma).
\]
Up to a change in universal constants,
\cref{thm:hsep-endpoint,cor:hsep-uniform-additive} remain valid after replacing
$\Sep(d:d)$, $h_{\Sep(d:d)}$, and the embedded pure product states by
$\Sep_{\mathrm{bos}}^{(2)}(d)$, $h_{\mathrm{bos}}^{(2)}$, and the states
$\proj\psi^{\otimes2}$, respectively.  \Cref{cor:sep-spectrahedral-shadow}
remains valid after replacing $\Sep(d:d)$ by
$\Sep_{\mathrm{bos}}^{(2)}(d)$.  The dimension bound
\eqref{eq:hsep-extension-complexity} remains valid in all three bosonic
settings.
\end{corollary}

\begin{proof}
$\sigma_S$ is a bosonic state.  Hence the same witnesses $W_T$ and product
states $\sigma_S$ can be used to construct the slack matrix and apply the
same arguments as in the nonbosonic case.  Therefore, all results carry over.
\end{proof}

\section{Approximate disentanglers as separability SDPs}
\label{sec:disentangler-reduction}

The following channel-to-SDP gadget comes from HNW
\cite[Sec.~5.3 and proof of Thm.~5.9]{harrowLimitationsSemidefinitePrograms2019}:
they enlarge the channel image and represent a Hermitian
correction as the difference of positive and negative PSD parts.  We retain
that gadget but make the corrected output an explicit PSD, trace-one state
$X$.

The $h_{\Sep}$ theorem already contains the PSD-rank lower bound.  It
remains only to show that an approximate disentangler supplies a small
objective-independent formulation of $h_{\Sep}$.  The idea is to optimize
$\Tr(QX)$ over all states $X$ within trace distance $\delta$ of some channel
output $\Lambda(\rho)$.  Every separable state is such an $X$ by the
covering condition \eqref{eq:disentangler-inner}, and every such $X$ is
within $\varepsilon+\delta$ of the separable set by the outer condition
\eqref{eq:disentangler-outer}.  The SDP below enforces
$\disttr(X,\Lambda(\rho))\le\delta$ through auxiliary blocks
$E_\pm\succeq0$ with $X-\Lambda(\rho)=E_+-E_-$ and
$\Tr(E_++E_-)\le2\delta$.

\begin{proposition}[A disentangler gives a uniform $h_{\Sep}$ formulation]
\label[proposition]{prop:channel-to-hsep}
Assume first that the two local output dimensions are both $d$.  An
$(\varepsilon,\delta)$-approximate disentangler
\[
 \Lambda:\Dens(\C^D)\longrightarrow\Dens(\C^d\otimes\C^d)
\]
induces an objective-independent SDP extended formulation of size at most
\begin{equation}
 r_\Lambda=2(D+3d^2+1)
 \label{eq:channel-hsep-size}
\end{equation}
whose optimum satisfies, for every $Q$ with $0\preceq Q\preceq I$,
\begin{equation}
 h_{\Sep(d:d)}(Q)
 \le\operatorname{SDP}_\Lambda(Q)
 \le h_{\Sep(d:d)}(Q)+\varepsilon+\delta.
 \label{eq:channel-uniform-hsep}
\end{equation}
Consequently, when $0<a\le1/2$, it is a
$(1/2+a,1/2)$-approximate HNW formulation with relative gap $2a$, where
$a=\varepsilon+\delta$.
\end{proposition}

\begin{proof}
Consider first the complex Hermitian block-diagonal spectrahedron of size
\[
 q=D+3d^2+1
\]
whose variables are
\[
 \rho\succeq0,
 \qquad E_+\succeq0,
 \qquad E_-\succeq0,
 \qquad X\succeq0,
 \qquad t\ge0,
\]
subject to the affine constraints
\begin{align}
 \Tr\rho&=1,
 &\Tr X&=1,
 \label{eq:channel-lift-traces}\\
 X&=\Lambda(\rho)+E_+-E_-,
 &\Tr(E_++E_-)+t&=2\delta.
 \label{eq:channel-lift-corrections}
\end{align}
It is nonempty: one may take $X=\Lambda(\rho)$, $E_+=E_-=0$, and
$t=2\delta$.  Applying the realification map from
\cref{lem:realification} to each Hermitian variable turns every complex
PSD constraint into a real PSD constraint of twice the size.  All remaining
conditions become linear equations in the entries of the realified matrices.
Hence the resulting feasible set is a real spectrahedron of size $2q$.
For a measurement operator $Q$ with $0\preceq Q\preceq I$, use the objective
\[
 \Phi_Q(\rho,E_+,E_-,X,t)=\Tr(QX).
\]
After realification this is still affine by
\eqref{eq:realification-trace}, and it lies in $[0,1]$ because $X$ is a
state.

We next construct the objective-independent embedding.  Let $\tau$ be a pure
product state.  By the covering property, choose $\rho_\tau$ satisfying
\[
 \disttr\bigl(\Lambda(\rho_\tau),\tau\bigr)\le\delta.
\]
Write the trace-zero Hermitian difference as
\[
 \tau-\Lambda(\rho_\tau)=E_{\tau,+}-E_{\tau,-}
\]
using its Jordan decomposition.  \Cref{eq:trace-zero-jordan} gives
\[
 \Tr(E_{\tau,+}+E_{\tau,-})
 =\norm{\tau-\Lambda(\rho_\tau)}_1
 \le2\delta.
\]
Set $X=\tau$ and
$t=2\delta-\Tr(E_{\tau,+}+E_{\tau,-})$.  This gives a feasible point whose
choice depends only on $\tau$, not on $Q$, and
$\Phi_Q=\Tr(Q\tau)$.

Conversely, let $(\rho,E_+,E_-,X,t)$ be feasible.  Since both $X$ and
$\Lambda(\rho)$ have trace one, $E_+-E_-$ has trace zero, and
\begin{align*}
 \disttr\bigl(X,\Lambda(\rho)\bigr)
 &=\frac12\norm{E_+-E_-}_1\\
 &\le\frac12\Tr(E_++E_-)
 \le\delta.
\end{align*}
The outer disentangler condition supplies a state $\sigma\in\Sep(d:d)$ such
that
\[
 \disttr(\Lambda(\rho),\sigma)\le\varepsilon.
\]
The triangle inequality therefore gives
\[
 \disttr(X,\Sep(d:d))\le\varepsilon+\delta=a.
\]
\Cref{eq:measurement-value-distance} now gives
\[
 \Tr(QX)\le h_{\Sep(d:d)}(Q)+a.
\]
Taking the supremum over feasible points proves the upper bound in
\eqref{eq:channel-uniform-hsep}; the embedded pure product states prove the
lower bound.  When $0<a\le1/2$ and $h_{\Sep(d:d)}(Q)\le1/2$, the SDP
optimum is at most $1/2+a$, so the midpoint thresholds give relative gap
$2a$.
\end{proof}

\begin{lemma}[Reduction to equal local dimensions]
\label[lemma]{lem:balanced-retraction}
Let $d=\min\{d_A,d_B\}$.  From every $(\varepsilon,\delta)$-approximate
disentangler with output dimensions $d_A,d_B$ and input dimension $D$, one
can obtain an $(\varepsilon,\delta)$-approximate disentangler with output
dimensions $d,d$ and the same input dimension.
\end{lemma}

\begin{proof}
Choose isometries
\[
 V_A:\C^d\longrightarrow\C^{d_A},
 \qquad
 V_B:\C^d\longrightarrow\C^{d_B},
\]
let $P_A=V_AV_A^*$ and $P_B=V_BV_B^*$, and fix states
$\omega_A,\omega_B\in\Dens(\C^d)$.  Define local channels
\begin{align*}
 \Phi_A(Z)&=V_A^*ZV_A+\Tr((I-P_A)Z)\,\omega_A,\\
 \Phi_B(Z)&=V_B^*ZV_B+\Tr((I-P_B)Z)\,\omega_B.
\end{align*}
These maps are completely positive and trace preserving.  On the chosen
$d$-dimensional subspaces, they invert the embeddings $V_A$ and $V_B$.  Set
\[
 \Lambda'=(\Phi_A\otimes\Phi_B)\circ\Lambda.
\]
Local channels preserve separability and contract trace distance, so the
outer error does not increase.  For the covering condition, embed a target
$\tau\in\Sep(\C^d:\C^d)$ as
\[
 \widetilde{\tau}
 =(V_A\otimes V_B)\tau(V_A\otimes V_B)^*
 \in\Sep(\C^{d_A}:\C^{d_B}).
\]
Choose an input $\rho$ with
$\disttr(\Lambda(\rho),\widetilde{\tau})\le\delta$.  Contractivity and the
identities $\Phi_A(V_A ZV_A^*)=Z$ and $\Phi_B(V_B ZV_B^*)=Z$ give
$\disttr(\Lambda'(\rho),\tau)\le\delta$.  Thus $\Lambda'$ has the claimed
parameters.
\end{proof}

\begin{proof}[Proof of \cref{thm:dimension}]
By \cref{lem:balanced-retraction}, it suffices to consider equal local
dimension $d$.  First assume $a>0$.  \Cref{prop:channel-to-hsep}
produces a $(1/2+a,1/2)$-approximate formulation of size
\[
 r_\Lambda=2(D+3d^2+1),
\]
whose relative gap is $\varrho=2a$.  Choose
$a_0\le\varrho_0/2$.  \Cref{thm:hsep-endpoint} then gives
\begin{equation}
 \log r_\Lambda
 \ge c_{\mathrm{end}}M_a\log(2+R_a).
 \label{eq:channel-before-absorption}
\end{equation}

It remains only to absorb the additive $O(d^2)$ lift overhead.  Put
$A_0=(2a_0)^{-1/3}$.  After increasing $d_0$, the second term in the
definition of $M_a$ is at least $A_0$, so $M_a\ge A_0$.
If $M_a\le N^{1/7}$, then
\[
 R_a\ge\frac{N^{1/2}}{(\log N)^3},
\]
and hence $\log(2+R_a)\ge(\log N)/3$ for large $N$.  If
$M_a>N^{1/7}$, then
\[
 M_a\log(2+R_a)\ge N^{1/7}\log3
 \ge\frac{A_0}{3}\log N
\]
for all sufficiently large $N$.  Thus in both cases
\begin{equation}
 M_a\log(2+R_a)\ge\frac{A_0}{3}\log N.
 \label{eq:channel-overhead-scale}
\end{equation}
Choose $a_0$ small enough that
$c_{\mathrm{end}}A_0/3\ge8$.  \Cref{eq:channel-before-absorption,eq:channel-overhead-scale}
imply $r_\Lambda\ge N^8$, which for large
$d$ is at least $12d^2+4$.  Since
$r_\Lambda=2D+6d^2+2$, we obtain
$D\ge r_\Lambda/4$.  Reducing the universal constant and increasing $d_0$
once more yields \eqref{eq:dimension-endpoint}.

If $a=0$, the channel also satisfies the disentangler conditions with any
positive total error $\widetilde{a}$.  Choose $\widetilde{a}$ so that
\[
 (2\widetilde{a})^{-1/3}
 \ge\left(\frac{N}{(\log N)^3}\right)^{2/7},
\]
apply the positive-error case, and obtain the same bound with
$(2a)^{-1/3}=+\infty$.
\end{proof}

\begin{proof}[Proof of \cref{cor:dimension-power}]
After the equal-dimension reduction, \cref{prop:channel-to-hsep}
and \eqref{eq:hsep-extension-complexity}, applied with effective error
$\eta=2a$, give
\[
 2(D+3d^2+1)
 \ge d^{\,c'_\theta\min\{(2a)^{-1/3},d^\theta\}}.
\]
Shrink $a_\theta$, enlarge $d_\theta$, and absorb the additive $O(d^2)$ term
exactly as in the proof of \cref{thm:dimension}.  Constant factors in
$(2a)^{-1/3}$ are absorbed into $c_\theta$.  This establishes
\eqref{eq:dimension-power}.
\end{proof}

\begin{corollary}[Input-dimension bounds in standard accuracy regimes]
\label[corollary]{cor:disentangler-regimes}
Let $d=2^n$.
\begin{enumerate}
 \item If $a\le(\log d)^{-c}=\Theta(n^{-c})$ for a fixed $c>0$, then
 \[
  \log_2D=\Omega(n^{1+c/3}),
  \qquad
  D\ge2^{\Omega(n^{1+c/3})}.
 \]
 \item If $a\le d^{-\beta}$ for a fixed $0<\beta<6/7$, then
 \[
  \log D=\Omega_\beta(d^{\beta/3}\log d).
 \]
 \item If
 \[
  a\le c_0\frac{(\log d)^{18/7}}{d^{6/7}}
 \]
 for a sufficiently small universal $c_0>0$---in particular, if
 $a\le d^{-\beta}$ with $\beta\ge6/7$---then
 \[
  \log D
  =\Omega\!\left(\frac{d^{2/7}}{(\log d)^{6/7}}\right).
 \]
\end{enumerate}
\end{corollary}

\begin{proof}
The first statement follows from \cref{cor:dimension-power} because
$a^{-1/3}\ge(\log d)^{c/3}$ and every fixed positive power of $d$ dominates
that quantity.  The second follows by choosing
$\theta\in(\beta/3,2/7)$.  For the third, the stated accuracy condition gives
\[
 (2a)^{-1/3}
 =\Omega\!\left(\left(\frac{d-1}{(\log(d-1))^3}\right)^{2/7}\right).
\]
Thus the second term defines $M_a$ up to universal constants.  Since
$R_a\ge1$, \cref{thm:dimension} gives the stated bound.
\end{proof}

\paragraph{Acknowledgements.}
SG and DR acknowledge support from the Deutsche Forschungsgemeinschaft
(DFG), project 563388236 (Bridge-QS, SPP 2514). SG additionally acknowledges
support from DFG project 572703436 (QPUP), EU QuantERA/DFG project 583918116
(SDPCODE), and BMFTR (PhoQuant).

\clearpage
\printbibliography

\end{document}